\PassOptionsToPackage{dvipsnames}{xcolor}
\pdfoutput=1

\newif\ifdraft\draftfalse 
\newif\ifanon\anonfalse    
\newif\iffull\fullfalse   
\newif\iflongrefs\longrefsfalse 
\newif\ifbackref\backreffalse 
\newif\ifsooner\soonerfalse
\newif\iflater\laterfalse
\newif\ifieee\ieeetrue
\newif\ifcamera\cameratrue 
\newif\ifappendix\appendixfalse 
\newif\ifallcites\allcitesfalse
\newif\ifneedspace\needspacetrue
\newif\ifhighlightnewtext\highlightnewtextfalse 
\newif\ifdiff\difffalse 
\newif\ifhighlightnewreviews\highlightnewreviewsfalse 
\newif\ifdiffreviews\diffreviewsfalse 
\newif\ifaggressivecomments\aggressivecommentsfalse 
\makeatletter
  \@input{texdirectives.tex} 
\makeatother

\ifieee
\documentclass[10pt, conference, compsocconf, letterpaper, times]{IEEEtran}

\DeclareMathAlphabet{\mathit}{\encodingdefault}{\familydefault}{m}{it}

\else 

\ifcamera
\documentclass[acmsmall]{acmart}
\else
\ifanon
\documentclass[acmsmall,review,anonymous,nonacm]{acmart}
\else
\documentclass[acmsmall,review=false,screen,nonacm]{acmart}
\fi
\fi

\acmJournal{PACMPL}
\acmYear{2018}
\acmMonth{1}
\acmVolume{2}
\aacmNumber{POPL}
\acmPrice{}
\acmArticle{65}
\acmMonth{1}
\acmDOI{10.1145/3158153}
\ifcamera\else\copyrightyear{2017}\fi 

\ifanon
\setcopyright{none}             
\else
\setcopyright{rightsretained}   
\fi

\makeatletter
\def\@copyrightpermission{\ifcamera\\\\\\\fi This work is licensed under a \href{https://creativecommons.org/licenses/by/4.0/}{Creative Commons Attribution 4.0 International License}}
\makeatother

\fi 

\usepackage{booktabs}   
\usepackage{subcaption} 

\usepackage[inference]{semantic}
\usepackage{amsmath}\allowdisplaybreaks
\usepackage{mathpartir}
\mprset {sep=1em} 
\usepackage{amsthm}
\usepackage{amssymb}
\usepackage{stmaryrd} 
\usepackage{xspace}
\usepackage{latexsym}
\usepackage{ifthen}
\usepackage{mathtools}
\usepackage{color}
\usepackage{listings}
\usepackage{verbatim}
\usepackage{tikz}
\usetikzlibrary{positioning,shadows,arrows,calc,backgrounds,fit,shapes,shapes.multipart,decorations.pathreplacing,shapes.misc,patterns,decorations.markings}
\usepackage{tikzscale}
\usepackage{tikz-qtree}
\usepackage[T1]{fontenc}
\usepackage[scaled=.83]{beramono}
\usepackage{epigraph}
\usepackage{booktabs}
\usepackage{float}
\usepackage{etoolbox}
\usepackage{wrapfig}
\usepackage{scalerel}
\usepackage{nameref}
\usepackage[strict]{changepage}
\usepackage{bm}
\usepackage{graphicx}
\usepackage[makeroom]{cancel}
\usepackage{centernot}
\usepackage{thm-restate}
\usepackage{enumitem}
\usepackage{fancyvrb}
\usepackage{varwidth}
\usepackage{combelow}
\usepackage{csquotes}

\usepackage{titlecaps}
\usepackage[normalem]{ulem}
\usepackage{cals}
\usepackage{afterpage}
\usepackage{placeins}

\ifieee
\usepackage[numbers,sort]{natbib}

\fi

\usepackage{chngcntr}

\ifanon
\usepackage[switch]{lineno}
\linenumbers

\fi

\usepackage{stfloats}

\definecolor{dkblue}{rgb}{0,0.1,0.5}
\definecolor{dkgreen}{rgb}{0,0.3,0}
\definecolor{dkred}{rgb}{0.6,0,0}
\definecolor{dkpurple}{rgb}{0.7,0,1.0}
\definecolor{purple}{rgb}{0.9,0,1.0}
\definecolor{olive}{rgb}{0.4, 0.4, 0.0}
\definecolor{teal}{rgb}{0.0,0.4,0.4}
\definecolor{azure}{rgb}{0.0, 0.5, 1.0}
\definecolor{gray}{rgb}{0.5, 0.5, 0.5}
\definecolor{dkgray}{rgb}{0.3, 0.3, 0.3}

\definecolor{cbgreen}{RGB}{000,158,115}
\definecolor{cbblue}{RGB}{000,114,178}
\definecolor{cbteal}{RGB}{091,142,253}
\definecolor{cbyellow}{RGB}{240,228,066}
\definecolor{cborange}{RGB}{230,140,000}
\definecolor{cbred}{RGB}{213,094,000}
\definecolor{cbpurple}{RGB}{204,121,167}

\newcommand*{\highlight}[1]{{\color{cbred} #1}}

\ifieee
\usepackage[
    pdftex,%
    pdfpagelabels,%
    colorlinks,%
    linkcolor=dkblue,%
    citecolor=dkblue,%
    filecolor=dkblue,%
    urlcolor=dkblue%
    \ifbackref,backref\fi%
]{hyperref}
\else
\usepackage{hyperref}
\hypersetup{
  breaklinks=true,
  citecolor=black,
  linkcolor=black,
  urlcolor=black,
}
\fi

\def\Snospace~{\S{}}

\def\Nnospace~{}

\usepackage{etoolbox}
\makeatletter
\patchcmd{\hyper@makecurrent}{%
    \ifx\Hy@param\Hy@chapterstring
        \let\Hy@param\Hy@chapapp
    \fi
}{%
    \iftoggle{inappendix}{
        \@checkappendixparam{chapter}%
        \@checkappendixparam{section}%
        \@checkappendixparam{subsection}%
        \@checkappendixparam{subsubsection}%
        \@checkappendixparam{paragraph}%
        \@checkappendixparam{subparagraph}%
    }{}%
}{}{\errmessage{failed to patch}}

\newcommand*{\@checkappendixparam}[1]{%
    \def\@checkappendixparamtmp{#1}%
    \ifx\Hy@param\@checkappendixparamtmp
        \let\Hy@param\Hy@appendixstring
    \fi
}
\makeatletter

\newtoggle{inappendix}
\togglefalse{inappendix}

\apptocmd{\appendix}{\toggletrue{inappendix}}{}{\errmessage{failed to patch}}


\AtBeginEnvironment{tikzpicture}{\catcode`$=3 } 

\newcommand{\comm}[3]{\ifdraft{{\color{#1}\ifaggressivecomments\bfseries \color{red}\fi[#2: #3]}}\fi}
\newcommand{\ch}[1]{\comm{teal}{CH}{#1}}
\newcommand{\jb}[1]{\comm{dkpurple}{JB}{#1}}
\newcommand{\yh}[1]{\comm{dkgray}{YH}{#1}}

\newcommand{\remove}[1]{\ifdiff{\color{RedOrange}\sout{#1}}\fi}

\usepackage{marginnote}
\newcommand{\mnote}[1]{\marginnote{\small#1}[-\baselineskip]}

\newcommand{\newtext}[1]{\ifhighlightnewtext{\color{MidnightBlue}#1}\else#1\fi}
\makeatletter
\newcommand{\newreviews}[2][\@nil]{%
  \def\tmp{#1}%
  \ifhighlightnewreviews{\color{ForestGreen}%
      \ifx\tmp\@nnil\else\mnote{#1}\fi%
  #2}\else#2\fi}
\newcommand{\removereviews}[2][\@nil]{%
  \def\tmp{#1}%
  \ifdiffreviews{\color{red}%
      \ifx\tmp\@nnil\else\mnote{#1}\fi%
  \sout{#2}}\fi}
\makeatother

\newcommand*{\EG}{e.g.,\xspace}
\newcommand*{\IE}{i.e.,\xspace}

\newcommand{\Triosecuris}{Triosecuris\xspace}

\newcommand*{\state}[1]{\langle #1 \rangle}
\newcommand*{\var}[1]{\mathtt{#1}}
\newcommand*{\MiniMIR}{\textsc{MiniMIR}\xspace}
\newcommand*{\MiniMC}{\textsc{MiniMC}}
\newcommand*{\cskip}{\texttt{skip}}
\newcommand*{\casgn}[2]{#1 ~ \texttt{:=} ~ #2}
\newcommand*{\cdiv}[3]{#1 ~ \texttt{:=} ~ #2 \texttt{/} #3}
\newcommand*{\cbranch}[2]{\texttt{branch}~#1~\texttt{to}~#2}
\newcommand*{\cjump}[1]{\texttt{jump}~#1}
\newcommand*{\cload}[2]{#1 \leftarrow \texttt{load[} #2 \texttt{]}}
\newcommand*{\cpeek}[1]{#1 \leftarrow \texttt{peek}}
\newcommand*{\cstore}[2]{\texttt{store[} #1 \texttt{]} \leftarrow #2}
\newcommand*{\ccall}[1]{\texttt{call}~#1}
\newcommand*{\cctarget}{\texttt{ctarget}}
\newcommand*{\cret}{\texttt{ret}}

\newcommand*{\ibranch}{\texttt{branch}}
\newcommand*{\ijump}{\texttt{jump}}

\newcommand*{\icall}{\texttt{call}}

\newcommand*{\ccond}[3]{#1 \operatorname{\texttt{?}} #2 \operatorname{\texttt{:}} #3}

\newcommand*{\st}{\mathit{st}}
\newcommand*{\msf}{\mathit{msf}}
\newcommand{\callee}{\mathit{callee}}
\newcommand*{\ms}{\mathit{ms}}
\newcommand*{\pc}{\mathit{pc}}
\newcommand*{\tpc}{\mathit{tpc}}
\newcommand*{\ct}{\mathit{ct}}
\newcommand*{\stk}{\mathit{stk}}
\newcommand*{\UV}{\texttt{UV}}

\newcommand*{\prog}{\texttt{p}}
\newcommand*{\m}{\texttt{m}} 
\newcommand*{\fetch}[3]{#1 \llbracket #2 \rrbracket = #3 }
\newcommand*{\mcfetch}[3]{#1 \texttt{[} #2 \texttt{]} = #3 }

\newcommand{\ite}[3]{\mathsf{if}\ #1\ \mathsf{then}\ #2\ \mathsf{else}\ #3}
\newcommand*{\mfalse}{\mathit{\bot}}
\newcommand*{\mtrue}{\mathit{\top}}

\newcommand{\lengthof}[1]{\left\lvert#1\right\lvert}

\newcommand*{\eval}[2]{\llbracket #1 \rrbracket_{#2}}
\newcommand*{\subst}[3]{[#1 \mapsto #2] #3}
\newcommand*{\cfg}[4]{#1,#2,#3,#4}
\newcommand*{\icfg}[2]{#1,#2}
\newcommand*{\scfg}[3]{#1,#2,#3}
\newcommand*{\runningstate}[1]{\langle #1\rangle}
\newcommand*{\faultstate}{\mathsf{Fault}}
\newcommand*{\termstate}{\mathsf{Term}}
\newcommand*{\emptylist}{\mathbin{[\,]}}
\newcommand*{\seqstep}[1]{\xrightarrow{#1}}
\newcommand*{\seqmulti}[1]{\xrightarrow{#1}{\hspace{-0.5em}}{}^*}

\newcommand*{\specstep}[2]{\xrightarrow[#2]{#1}\hspace{-0.5em}{}_s}
\newcommand*{\specmulti}[2]{\xrightarrow[#2]{#1}\hspace{-0.5em}{}_s^*}

\newcommand*{\idealstep}[2]{\xrightarrow[#2]{#1}\hspace{-0.5em}{}_i}
\newcommand*{\idealmulti}[2]{\xrightarrow[#2]{#1}{\hspace{-0.5em}}{}_i^*}
\newcommand*{\concat}{\mathbin{+\!\!+}}
\newcommand{\kw}[1]{\mathsf{#1}}

\newcommand{\barB}{\boldsymbol{|}}

\newcommand{\op}{\mathrm{op}}
\newcommand*{\wfp}[1]{\mathit{wf\_prog} ( #1 )}
\newcommand*{\wwfp}[1]{\mathit{wwf\_prog} ( #1 )} 
\newcommand*{\wfret}[1]{\mathit{wf\_ret} ( #1 )}
\newcommand*{\wfreg}[1]{\mathit{wf\_reg} ( #1 )}
\newcommand*{\wfmem}[1]{\mathit{wf\_mem} ( #1 )}
\newcommand*{\len}{\mathit{len}}
\newcommand*{\wfds}[2]{\mathit{wf\_ds} \; #2 \; #1 }
\newcommand*{\wfdsmc}[3]{\mathit{wf\_ds}_{\texttt{mc}} \; #3 \; #2 \; #1 }

\newcommand*{\UsedVars}[1]{\mathit{VARS} ( #1 )}
\newcommand*{\isct}[1]{\mathit{IsCallTarget} ( #1 )} 
\newcommand*{\issafe}[3]{\mathit{Safe}_{\mathit{#2}} \; #3  \; #1 }
\newcommand*{\efp}[1]{\texttt{\&}#1} 
\newcommand*{\vfp}[1]{\mathcal{\scriptstyle\&}#1} 
\newcommand{\seqobseq}{\approx}
\newcommand{\specobseq}{\approx_\mathit{s}}
\newcommand{\idealobseq}{\approx_\mathit{i}}
\newcommand{\mcobseq}{\approx^{\texttt{mc}}_{\mathit{s}}}

\newcommand*{\Obss}{\mathcal{O}}
\newcommand*{\Dirs}{\mathcal{D}}
\newcommand*{\obranch}[1]{\textit{OBranch}\,#1}
\newcommand*{\dbranch}[1]{\textit{DBranch}\,#1}
\newcommand*{\ocall}[1]{\textit{OCall}\,#1}
\newcommand*{\dcall}[1]{\textit{DCall}\,#1}
\newcommand*{\dret}[1]{\textit{DRet}\,#1}
\newcommand*{\oload}[1]{\textit{OLoad}\,#1}
\newcommand*{\ostore}[1]{\textit{OStore}\,#1}
\newcommand*{\odiv}[2]{\textit{ODiv}\,#1\,#2}

\newcommand{\tr}[1]{\llparenthesis \, #1 \, \rrparenthesis}
\newcommand{\mc}[1]{\llparenthesis \, #1 \, \rrparenthesis^{\texttt{mc}}}

\theoremstyle{definition}
\newtheorem{example}{Example}
\newtheorem{definition}{Definition}
\newtheorem{theorem}{Theorem}
\newtheorem{lemma}{Lemma}

\newtheorem{listing}{Listing}

\ifieee\else\fi

\newcommand{\titleString}{\Triosecuris: Formally Verified Protection Against\\Speculative Control-Flow Hijacking}
\title{\titleString}

\ifneedspace
\title{\huge\bf{\titleString}
\ifneedspace\ifieee\ifanon\vspace{-1.5em}\else\vspace{-0.0em}\fi\fi\else\ifanon\vspace{2em}\fi\fi}
\else
\title{\titleString}
\subtitle{\subtitleString}
\fi

\ifanon
\author{}
\else

\ifieee\large
\author{
  Jonathan Baumann\textsuperscript{1\ddag} \quad 
  Yonghyun Kim\textsuperscript{1\ddag}\quad
  Yan Farba\textsuperscript{1,2}\quad
  C\u{a}t\u{a}lin Hri\cb{t}cu\textsuperscript{1} \quad 
  Julay Leatherman-Brooks\textsuperscript{1,3}
\\[0.5em]
{\small
  \textsuperscript{1}MPI-SP,\ifcamera\else{} Bochum,\fi{} Germany\quad
  \textsuperscript{2}Ruhr University Bochum, Germany\quad
  \textsuperscript{3}Portland State University, USA}\\[0em]
}
\else 

\author{AUTHOR1}
\affiliation{
  \ifcamera\institution{AFF1}\city{CITY}\country{COUNTRY}
  \else\institution{AFF1}\fi}
\email{EMAIL}

\makeatletter
\renewcommand{\@shortauthors}{SHORTAUTHORS}
\makeatother

\fi 

\fi 

\begin{document}


\ifieee\maketitle\fi

\ifanon\else
{\renewcommand{\thefootnote}{\fnsymbol{footnote}}
\footnotetext[3]{Equal contribution}}
\fi


\begin{abstract}
This paper introduces \removereviews{SpecIBT}\newreviews[REQ2]{\Triosecuris}, a formally verified defense against Spectre BTB, RSB, and PHT that combines CET-style hardware-assisted control-flow integrity with compiler-inserted speculative load hardening (SLH).
\Triosecuris is based on the novel observation that in the presence of CET-style
protection, we can precisely detect BTB misspeculation for indirect calls
\newreviews[REQ2]{and RSB misspeculation for returns} and set the SLH misspeculation flag.
We formalize \Triosecuris as a transformation in Rocq and provide a machine-checked proof that it achieves relative security: any transformed program running with speculation leaks no more than what the source program leaks without speculation.
This strong security guarantee applies to arbitrary programs, even those not following the cryptographic constant-time programming discipline.

\end{abstract}

\ifieee
\renewcommand\IEEEkeywordsname{Keywords}
\fi

\section{Introduction}
\label{sec:intro}

The Spectre speculative side-channel attacks are a formidable threat for the
security of software systems~\cite{KocherHFGGHHLM019, CanellaBSLBOPEG19, HertoghQRSMB26}.
Without any defenses, attackers can easily construct ``universal read gadgets''
that leak a sensitive program's entire memory~\cite{KocherHFGGHHLM019, McIlroySTTV19}.
Yet strong and efficient Spectre defenses are notoriously challenging to implement.
Even the Linux kernel, for which many developers have spent years to implement
effective Spectre defenses, still sometimes falls prey to Spectre attacks, including
universal read gadgets~\cite{WiebingTBG24, WiebingG25}, and the situation
seems worse when developers have fewer resources and rely
entirely on the defenses currently provided by the OSes and compilers.
These defenses often involve a complex combination of software in compilers and
OSes as well as hardware features and models, which are all very difficult to get right.
Often the attacks circumventing Spectre defenses exploit implementation
mistakes, bad security-efficiency tradeoffs, unclear attacker models, and
unclear and overly optimistic assumptions.

Two particularly powerful Spectre variants are Spectre BTB~\cite{KocherHFGGHHLM019}, which exploits
indirect branch target speculation to jump to arbitrary addresses, and Spectre
RSB~\cite{KoruyehKSA18, MaisuradzeR18}, which exploits return address speculation.
%
These variants can be used to speculatively hijack control flow and then exploit
arbitrary gadgets in the code to leak secrets~\cite{BhattacharyyaSK20}.
Low-overhead mitigations against these variants can benefit from existing
control-flow integrity (CFI)~\cite{AbadiBEL09, BurowCNLFBP17}
features in modern processors.
In particular, Intel CET~\cite{CET} supports CFI protection both for forward
control-flow edges (e.g. indirect jumps) via Indirect Branch Tracking (IBT)
and also for backward edges (shadow stack for return addresses).
These hardware CFI features were originally introduced to defend against hijacking
control-flow using memory corruption~\cite{CET, GaidisMSMAK23},
but were later also applied to mitigate
speculative control-flow hijacking~\cite{GaidisMSMAK23, KoruyehSKSA20,
  MosierNMT24, NarayanDMCJGVSS21}.
%
In response to such defenses~\cite{MosierNMT24}, Intel updated its CET
documentation~\cite{CETspec}, confirming that the CET implementation in some of its current
processors (including products like Alder Lake-N and Arizona Beach) and the
intended long-term direction for CET is to also provide security under speculation.

Yet some of this previous work using CET~\cite{GaidisMSMAK23,KoruyehSKSA20}
disregards that Spectre has a different attacker model than memory corruption.
In particular, even the most fine-grained forward-edge CFI defenses against
memory corruption~\cite{GaidisMSMAK23} have to statically {\em overapproximate} the
set of possible targets of indirect calls and then dynamically check that each
call goes to a target in the corresponding set.
Directly applying this technique to defend against Spectre BTB preserves this
approximation aspect~\cite{GaidisMSMAK23, KoruyehSKSA20}, which allows clever
attackers to still circumvent these defenses in practice~\cite{WiebingTBG24}.
%
%
Even Serberus~\cite{MosierNMT24}, which is
\newreviews[REQ3]{a complete defense against Spectre BTB (and RSB)}\removereviews{designed
to fully prevent Spectre BTB attacks}, does
not avoid overapproximation, which in the end leads to complex
enforcement (\EG procedure-private stacks, register cleaning)
and limitations (\EG secrets can't be passed by value).

\newreviews[REQ2]{
Moreover,
the CET shadow stack does not provide
perfect protection against Spectre RSB, even when also setting the
{\tt RRSBA\_DIS\_U} flag, which disables return address prediction using the BTB.
As noticed by the authors of Serberus~\cite{MosierNMT24}, some Spectre RSB
misspeculation is still possible, with returns being able to
speculatively jump to the instruction immediately following an arbitrary program call.
This overapproximation also adds to the complexity and limitations of Serberus.}

In this work, we show that no approximation is needed when using CET to
defend against Spectre BTB \newreviews[REQ2]{and RSB}.
We do this by introducing \Triosecuris: a simple \removereviews[REQ2]{forward-edge  }CFI protection
specifically designed against the Spectre BTB \newreviews[REQ2]{and RSB} attacker
model\newtext{s}.
With \Triosecuris, the caller saves the function pointer
it is about to indirectly call in a
register and the callee checks that the value in this register is the correct one.
This check is useless in the absence of speculation, but it {\em precisely} detects
misspeculation, in which case we set a misspeculation flag register that
is used as a mask by Speculative Load Hardening (SLH).
\newreviews[REQ2]{For returns, we use the same kind of check, but with the return
  address on the stack, to detect RSB misspeculation and again set the SLH flag.}
%
\Triosecuris currently relies on Ultimate SLH~\cite{ZhangBCSY23},
a Spectre PHT defense (very similar to strong SLH~\cite{PatrignaniG21})
using the misspeculation flag to mask not only the values loaded from
memory, but also any operands that could be leaked via timing (\IE branch
conditions, addresses of memory accesses, etc.).


\paragraph{Contributions}
\begin{itemize}[leftmargin=*,nosep,label=$\blacktriangleright$]
\item
  We introduce \Triosecuris, a strong defense that protects arbitrary programs against
  Spectre BTB, RSB, and PHT by combining hardware-assisted control-flow integrity
  in the style of Intel CET with compiler-inserted Ultimate SLH. The main idea
  behind \Triosecuris is simple, but as far as we know novel: in the presence of
  CET-style protection, we can precisely detect BTB \newreviews[REQ2]{and RSB}
  misspeculation and set the SLH misspeculation flag.
  \newreviews[REQ4B1]{Moreover, since \Triosecuris is a conservative extension
  of Ultimate SLH, it still provides complete protection against
    Spectre PHT.}\footnote{\newreviews[REQ2]{We named our defense
      \Triosecuris because it provides security against three Spectre variants
      (BTB, RSB, and PHT), and also inspired by the \emph{Dioscuri} of Greek
      mythology, twins (our BTB and RSB protections) who guide
      travelers lost at sea (misspeculating programs) back to safety.}}
%
\item
  We formalize \Triosecuris in the Rocq prover\footnote{The
    \href{https://rocq-prover.org/}{Rocq interactive theorem prover} was
    previously known as Coq.} as a compiler from a simple low-level language
  with basic blocks and indirect function calls to a similar language with
  an idealized model of CET
  and a speculative semantics modeling the remaining capabilities
  of the attacker, who can still mount Spectre BTB, \newreviews[REQ2]{RSB,} and PHT
  attacks\removereviews{, but not Spectre RSB ones (which we assume can be stopped by a CET-style shadow stack)}.
  This mechanized formalization required getting the formal details right:
  For a start, transforming the code in the presence of indirect jumps requires
  a notion of basic blocks and an abstract notion of code pointers.
%
  So the source and target language of \Triosecuris feature CompCert-style function
  pointer values~\cite{CompCertMM, ZakowskiBYZZZ21},
  and passing an integer value to operations needing a code
  pointer (\IE indirect calls) and vice versa is undefined behavior.
%
  More subtly, in this setting the Ultimate SLH masking can change abstract code
  pointers into integers. Thus, in order to show that \Triosecuris does not introduce speculative
  undefined behavior, we also had to add to the semantics CompCert-style
  undefined values~\cite{CompCertMM, LeeKSHDMRL17}, so a more precise
  model of undefined behavior than that of, \EG the C standard.
\item
  We prove in Rocq that the \Triosecuris transformation achieves relative security:
  any transformed program running with speculation does not leak more than what
  the source program leaks sequentially~\cite{BaumannBDHH25}.
  This strong guarantee applies to arbitrary programs, even those not
  following the (restrictive) cryptographic constant-time discipline.
\item
  The target language of the \Triosecuris transformation above still includes basic
  blocks and abstract code pointers, so we also formalize a subsequent linearization
  step targeting machine code for a simple machine.
  This linearization step links basic blocks together, lays them out in memory, and
  replaces code pointers with concrete integers.
  We prove in Rocq relative security for this linearization step,
  on all safe inputs, which do not trigger undefined behavior.
\item
  We compose the proofs above to obtain a machine-checked end-to-end
  relative security proof, as illustrated in \autoref{fig:composition}.
  For this we additionally show that our \Triosecuris transformation is safety
  preserving, \IE inputs that do not trigger undefined behavior in the source also
  do not trigger it after the \Triosecuris step; this proof relies on undefined values.
\end{itemize}

\begin{figure}
  \centering
  \begin{tikzpicture}
    \node[draw, rectangle] (src) at (0,0) {
      \begin{minipage}{.5\linewidth}
        \centering
        \MiniMIR \\
        \footnotesize
        basic blocks, func. ptrs, undef. values,
        sequential semantics
      \end{minipage}
    };
    \node[draw, rectangle] (tgt) at (0,-1.7) {
      \begin{minipage}{.5\linewidth}
        \centering
        \MiniMIR \\
        \footnotesize
        basic blocks, func. ptrs, undef. values,
        speculative semantics + CET
      \end{minipage}
    };
    \node[draw, rectangle] (mc) at (0,-3.2) {
      \begin{minipage}{.5\linewidth}
        \centering
        \MiniMC~\footnotesize machine code \\
        speculative semantics + CET
      \end{minipage}
    };
    
    \draw[->] (src) -- (3, 0) -- node[right] {
      \begin{minipage}{.18\linewidth}
        \scriptsize 
        relative security\\ +
safety\\ preservation
      \end{minipage}
    } (3, -1.65) --
    ([shift=({0,0.05})]tgt.east);
    \draw[->] ([shift=({0,-0.05})]tgt.east) -- (3, -1.75) -- node[right] {
      \begin{minipage}{.12\linewidth}
        \scriptsize relative security\\
        for safe\\ inputs
      \end{minipage}
    } (3, -3.2) -- (mc);
    \draw[->] ([shift=({0,0.2})]src.east) -- (5.2, 0.2) -- node[right] {
      \begin{minipage}{.2\linewidth}
        \footnotesize 
        relative security\\
        for safe\\ inputs
      \end{minipage}
    }(5.2, -3.4) -- ([shift=({0,-0.2})]mc.east);

    \draw[->] (src) -- node[right] {\footnotesize \Triosecuris} (tgt);
    \draw[->] (tgt) -- node[right] {\footnotesize Linearization} (mc);
  
  \end{tikzpicture}
  \caption{Languages and security results}
  \label{fig:composition}
  \vspace{-1em}
\end{figure}

\newreviews[REQ1]{
\paragraph{Scope and Limitations}
In this work we formally investigate Triosecuris on simple languages that
capture the main conceptual ideas such as (1) detecting BTB and RSB misspeculation,
(2) transforming basic blocks to add such detection code on top of the masking and
\newreviews[REQ4B1]{Spectre PHT detection} of Ultimate SLH, (3) dealing with refinement of undefined
behavior and supporting undefined values.
We have not implemented these ideas in a real compiler like LLVM.
This is interesting future work though, which would allow experimentally
evaluating the overhead of Triosecuris, as further discussed in \autoref{sec:conclusion}.}

\paragraph{Outline} \autoref{sec:background} gives background,
after which \autoref{sec:key-ideas} introduces our key ideas.
\autoref{sec:defining-specibt} formally defines the \Triosecuris transformation
and \autoref{sec:formal-results} presents our security proof for it.
\autoref{sec:translation-mc} presents our translation to machine code and
our end-to-end security result.
\autoref{sec:related-work} discusses related work and
\autoref{sec:future-work} concludes with future work.
\ifappendix
The appendices include additional technical details.
 \else
 An extended version including appendices with additional technical
 details is available at \url{https://arxiv.org/abs/2601.22978}.
\fi

\paragraph{Artifact} The results of this paper have been fully formalized in the Rocq proof assistant.
The development leading to our main theorems has $\sim$\removereviews{5400}\newreviews[REQ2]{6800} lines of Rocq code
and is available
\ifanon
as anonymous supplementary material.
\else
at \mbox{\url{https://github.com/secure-compilation/Triosecuris}}.
\fi{}
The supplementary material also includes $\sim$\removereviews{2300}\newreviews[REQ2]{2500} lines of property-based
testing\ifanon\else~\cite{Paraskevopoulou15,testing_ni_jfp}\fi{} code that we used to
check our properties before proving them.
Testing helped discover and fix issues early, for instance it helped us
identify the need for adding undefined values to obtain safety preservation
(\autoref{sec:ki-undefined-values}).
\newtext{The authors used generative AI (Claude Opus, Microsoft Copilot, and Google Gemini) to assist with the Rocq mechanization of Spectre RSB protection,
refactoring, and minor coding assistance.
All AI-assisted development was reviewed by the authors, and all proofs mechanically checked by Rocq;
the authors take full responsibility for the entire artifact development.}

\section{Background}
\label{sec:background}

\subsection{Spectre PHT}
Nearly all modern processors make use of a variety of hardware mechanisms
to improve performance (\EG caches). These are part of the processor's
{\em microarchitecture} and are not
exposed in the {\em instruction set architecture} (ISA), so they should not affect the
results of computations. Still, their effects can be detected
across processes or even remotely by measuring execution time,
which can be used to leak secrets.

\emph{Speculative execution} is the result of a variety of techniques to avoid
stalling by executing code that is likely, but not guaranteed to be
needed by the program. While speculatively executed instructions do not affect
the \emph{architectural} state, they still change the
\emph{microarchitectural} state. This enables the well-known
\emph{Spectre} attacks~\cite{KocherHFGGHHLM019},
which exploit speculative execution to leak
secrets into the microarchitectural state, from which they are
observed via traditional timing side-channels.

Spectre attacks are classified by the 
prediction mechanism
triggering speculation~\cite{CanellaBSLBOPEG19}.
The first and most well-known variant, Spectre PHT~\cite{KocherHFGGHHLM019},
targets the \emph{pattern history table}, predicting the outcomes of
conditional branch instructions.

\begin{listing}[Vulnerable program]
  \label{lst:running-example}
  ~

  \centering
  \begin{tikzpicture}
    \node[draw, rectangle, scale=0.8] (callnum) {
      \begin{minipage}{0.55\linewidth}
        \vspace{-1em}
        \begin{align*}
          \mathit{calln}: \,& \cbranch{\var{arg1} < \var{len}}{l_\mtrue}\\
             & \casgn{\var{fun}}{\efp{\mathit{fun\_1}}}\\
             & \cjump{l_\mathit{cont}}
        \end{align*}
      \end{minipage}
    };
    \node[draw, rectangle, scale=0.8, below=0.1 of callnum] (ltrue) {
      \begin{minipage}{0.55\linewidth}
        \vspace{-1em}
        \begin{align*}
          l_\mtrue: \quad& \casgn{\var{fun}}{\efp{\mathit{fun\_2}}}\\
             & \cjump{l_\mathit{cont}}
        \end{align*}
      \end{minipage}
    };
    \node[draw, rectangle, scale=0.8, below=0.1 of ltrue] (lcont) {
      \begin{minipage}{0.55\linewidth}
        \vspace{-1em}
        \begin{align*}
          l_\mathit{cont}: \quad& \ccall \var{fun}\\
                                & \vspace{-.1em}\scalebox{0.8}{\large\vdots}
        \end{align*}
      \end{minipage}
    };
    \node[draw, rectangle, scale=0.8, right=0.2 of callnum] (fun1) {
      \begin{minipage}{0.6\linewidth}
        \vspace{-1em}
        \begin{align*}
          \mathit{fun\_1}: \quad& \dots \\
                                & \vspace{-.1em}\scalebox{0.8}{\large\vdots}
        \end{align*}
      \end{minipage}
    };
    \node[draw, rectangle, scale=0.8, below=0.1 of fun1] (fun2) {
      \begin{minipage}{0.6\linewidth}
        \vspace{-1em}
        \begin{align*}
          \mathit{fun\_2}: \quad& \cload{\var{x}}{\var{base} + \var{arg1}}\\
                                & \cload{\var{y}}{\var{x}}\\
                                & \vspace{-.1em}\scalebox{0.8}{\large\vdots}
        \end{align*}
      \end{minipage}
    };
  \end{tikzpicture}
\end{listing}
\begin{example}
  \label{ex:spectre-pht}
  Consider 
  \autoref{lst:running-example}. The load instructions
  at label $\mathit{fun\_2}$ affect the cache, allowing an attacker to
  determine which address was accessed via a timing attack. 
  Sequentially, this code is protected by the conditional branch at label
  $\mathit{calln}$,
  which ensures that $\mathit{fun\_2}$ will only be
  called if $\var{arg1}$ is within a certain range (\EG a public array starting
  at $\var{base}$).
  Thus, while the second load allows an attacker to determine the value loaded
  by the first, those values are guaranteed to be public during sequential
  execution.

  With a Spectre PHT attack, however, an attacker can train the branch predictor
  to predict that the branch will be taken,
  even for an out-of-bounds value for $\var{arg1}$.
  This would cause $\mathit{fun\_2}$ to be called speculatively with an
  out-of-bounds index, exposing an arbitrary memory location to the attacker,
  which is known as a universal read gadget.
\end{example}

\subsection{Ultimate SLH}
\label{sec:uslh}
One mitigation against Spectre PHT attacks is Ultimate SLH~\cite{ZhangBCSY23}
(very similar to a previous variant called strong SLH~\cite{PatrignaniG21}),
which builds upon the SLH implementation in LLVM~\cite{Carruth18}.
All SLH mitigations share the core concept of using a designated
\emph{misspeculation flag} register to detect misspeculation. They achieve this
by using branchless, unpredicted logic, which we represent in our setting with
the conditional expression $\ccond{e_\mathit{cond}}{e_\mtrue}{e_\mfalse}$,
to update the misspeculation flag after every branch as shown in
\autoref{lst:tracking-pht}.

\begin{listing}[Tracking PHT misspeculation]
  \label{lst:tracking-pht}
  We show only the basic blocks which are modified w.r.t.
  \autoref{lst:running-example}.

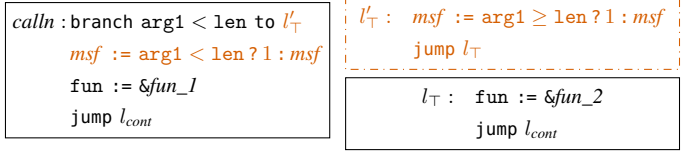
  \centering
  \begin{tikzpicture}
    \node[draw, rectangle, scale=0.8] (callnum) {
      \begin{minipage}{0.55\linewidth}
        \vspace{-1em}
        \begin{align*}
          \mathit{calln}: \,& \cbranch{\var{arg1} < \var{len}}{\highlight{l'_\mtrue}}\\
                  & \highlight{\casgn{\msf}{\ccond{\var{arg1} <
                  \var{len}}{1}{\msf}}}\\
             & \casgn{\var{fun}}{\efp{\mathit{fun\_1}}}\\
             & \cjump{l_\mathit{cont}}
        \end{align*}
      \end{minipage}
    };
    \node[draw, dashdotted, rectangle, scale=0.8, above right = -.9 and 0.15 of callnum, color=cbred] (ltrueprime) {
      \begin{minipage}{0.58\linewidth}
        \vspace{-1em}
        \begin{align*}
          l'_\mtrue: \quad& \casgn{\msf}{\ccond{\var{arg1} \geq
          \var{len}}{1}{\msf}}\\
                         & \cjump{l_\mtrue}
        \end{align*}
      \end{minipage}
    };
    \node[draw, rectangle, scale=0.8, below=0.1 of ltrueprime] (ltrue) {
      \begin{minipage}{0.58\linewidth}
        \vspace{-1em}
        \begin{align*}
          l_\mtrue: \quad& \casgn{\var{fun}}{\efp{\mathit{fun\_2}}}\\
             & \cjump{l_\mathit{cont}}
        \end{align*}
      \end{minipage}
    };
  \end{tikzpicture}
\end{listing}

This misspeculation flag is used to \emph{mask} vulnerable data, by once
again using branchless, unpredicted operations (in practice, bitwise boolean
operations are often used) to replace them with a safe value (in our case 0)
if the misspeculation flag is set. This has no effect
during sequential execution, but prevents these values from leaking due to
misspeculation.

\begin{listing}[Masking during misspeculation]
  \label{lst:running-example-masking}
  We show only the basic blocks which are modified w.r.t.
  \autoref{lst:running-example}. The misspeculation flag is set as in
  \autoref{lst:tracking-pht}.

  \centering
  \begin{tikzpicture}
    \node[draw, rectangle, scale=0.8] (fun2) {
      \begin{minipage}{0.6\linewidth}
        \vspace{-1em}
        \begin{align*}
          \mathit{fun\_2}: \quad& \cload{\var{x}}{\highlight{ \msf
          \operatorname{\texttt{?}} 0 \operatorname{\texttt{:}}} \var{base} +
        \var{arg1}}\\
                                & \cload{\var{y}}{\highlight{ \msf
          \operatorname{\texttt{?}} 0 \operatorname{\texttt{:}}}\var{x}}\\
                                & \vspace{-.1em}\scalebox{0.8}{\large\vdots}
        \end{align*}
      \end{minipage}
    };
  \end{tikzpicture}
\end{listing}

\begin{example}
  Recall the Spectre PHT attack described in \autoref{ex:spectre-pht}.
  With Ultimate SLH applied, this attack is no longer possible:
  The mispredicted branch is detected by the tracking inserted in
  \autoref{lst:tracking-pht} and the misspeculation flag is set, causing the
  addresses to be masked as shown in \autoref{lst:running-example-masking}.
  Now, even if these loads are executed speculatively, the attacker will only
  observe the constant address 0.
\end{example}

Ultimate SLH derives its name from its thorough approach to masking, which
includes not only loads, as shown above,
but also any operands that could be leaked via timing (\IE branch
conditions, addresses of memory accesses, leaked arguments of variable-time
instructions \newtext{such as division}, etc.).

\subsection{Relative Security}
\label{sec:bg-rs}
Several previous 
mitigations~\cite{ShivakumarBBCCGOSSY23,
OlmosBCGLOSYZ25, MosierNMT24} have restricted their scope to 
cryptographic code, which follows very strict disciplines such as
\emph{cryptographic constant-time} to ensure security against sequential
attackers~\cite{DanielBR21, DanielBR23, BartheBGHLPT20}.
In this setting, it is possible to prove that the hardened program
has no leakage sequentially or speculatively.

Instead of requiring constant-time,
our \Triosecuris transformation protects \emph{arbitrary} source programs,
as previously done by Ultimate SLH~\cite{ZhangBCSY23, PatrignaniG21}
and FSLH~\cite{BaumannBDHH25}.
As such, we can only hope to prove that the
hardened program \emph{does not leak more} under speculative execution than the
source program leaks sequentially. Specifically, if a non-speculative attacker cannot
distinguish between two input states, our mitigation should ensure that a
speculative attacker cannot do so either.
This is a relative notion of security~\cite{CauligiDMBS22,
  DongolGPW24, CheangRSS19, GuarnieriKMRS20}, as the hardened program is secure
\emph{relative} to the source program~\cite{BaumannBDHH25}.
Ultimate SLH satisfies relative security against Spectre PHT attacks~\cite{BaumannBDHH25}.

\subsection{Spectre BTB \newreviews[REQ2]{and RSB}}
\label{sec:bg-btb-rsb}

In this paper, we consider the more powerful Spectre
BTB~\cite{KocherHFGGHHLM019} \newreviews[REQ2]{and RSB~\cite{KoruyehKSA18, MaisuradzeR18}}
attacks.
\newreviews{In the former, }\removereviews{in which} the target location of an indirect call or jump is predicted
by the \emph{branch target buffer} (BTB), and the processor speculatively
executes code at the predicted location. A Spectre BTB attacker can train the
BTB to predict any address in the victim program, and can therefore
speculatively redirect control flow arbitrarily.
This can circumvent SLH protections completely:
\begin{example}
  \label{ex:btb-slh-bypass}
  Recall the Ultimate SLH-hardened program of \autoref{lst:tracking-pht} and
  \autoref{lst:running-example-masking}. An attacker with control of the BTB can
  let the first branch proceed correctly, but then misdirect the call in basic block
  $l_\mathit{cont}$ to $l_\mtrue + 1$, thus executing that branch without
  updating the $\msf$.
  Alternatively, the attacker could directly misdirect the call to
  $\mathit{fun\_2}$ when it should be $\mathit{fun\_1}$. Both attacks result in
  $\mathit{fun\_2}$ being executed speculatively without the misspeculation flag
  set, and thus exposing an arbitrary memory location to the attacker.
\end{example}

\newreviews[REQ2]{
In a Spectre RSB attack, the target location of a return is predicted by the
\emph{return stack buffer} (RSB), which keeps track of addresses of previous
call instructions. This buffer can also be manipulated by an
attacker~\cite{KoruyehKSA18}, however, only call instructions can add entries to
the RSB.\footnote{This assumes that the kernel fills the RSB on context switch,
clearing any entries from other processes.}
Return address prediction can often bypass the RSB, and use the
BTB instead~\cite{WiknerR22}. However, this behavior can be disabled using the
RRSBA\_DIS\_U flag\jb{cn}, which we require in our setting.
Thus, like Serberus~\cite{MosierNMT24}, we assume that the RSB can only predict locations
immediately following call instructions.
}

\subsection{Hardware-Assisted CFI Protections}
\label{sec:bg-hardware-cfi}
Modern processors offer hardware protections for forward-edge CFI such as
Intel's Indirect Branch Tracking (IBT)~\cite{CET} and ARM's Branch Target
Identification (BTI)~\cite{arm_bti}.
Compilers using these mechanisms
introduce a special instruction to mark valid targets for indirect branches
and calls, \EG the \texttt{ENDBR64} and \texttt{ENDBR32} instructions for
Intel's IBT~\cite{CET}. When IBT is enabled, this
marker must be the first instruction executed after an indirect call
or branch, otherwise a hardware fault is triggered and execution halted.

\begin{figure}
  \begin{subfigure}{\linewidth}
    \begin{tikzpicture}[
      scale=1.5,
      code/.style={ellipse, draw=black}
      ]
      \node[code] (call1) at (1, 0) {\texttt{calln}};
      \node[code] (call2) at (3.5, 0) {\texttt{callalpha}};

      \node[code] (tgt1) at (0, -1.5) {\texttt{fun\_1}};
      \node[code] (tgt2) at (1, -2) {\texttt{fun\_2}};
      \node[code] (tgta) at (2.5, -1.5) {\texttt{fun\_a}};
      \node[code] (tgtb) at (3.5, -2) {\texttt{fun\_b}};
      \node[code] (tgtc) at (4.5, -1.5) {\texttt{fun\_c}};

      \draw[->] (call1) [out=270, in=90] to (tgt1);
      \draw[->,dashed] (call1) [out=270, in=90] to (tgt2);
      \draw[dotted] (call1) [out=270, in=90] to (tgta);
      \draw[dotted] (call1) [out=270, in=90] to (tgtb);
      \draw[dotted] (call1) [out=270, in=90] to (tgtc);
      \draw[dotted] (call2) [out=270, in=90] to (tgt1);
      \draw[dotted] (call2) [out=270, in=90] to (tgt2);
      \draw[->,dashed] (call2) [out=270, in=90] to (tgta);
      \draw[->] (call2) [out=270, in=90] to (tgtb);
      \draw[->,dashed] (call2) [out=270, in=90] to (tgtc);
    \end{tikzpicture}
    \caption{Coarse-grained IBT~\cite{MosierNMT24, NarayanDMCJGVSS21}:
      calls are guaranteed to lead to function
      entry points, with no further restriction on the targeted function.
      This includes functions that should never be reached sequentially, here
      represented by dotted lines, as well as functions that can be reached sequentially,
      but are not the intended target (dashed lines).\vspace{1em}
}
    \label{fig:IBT}
  \end{subfigure}
  \begin{subfigure}{\linewidth}
    \begin{tikzpicture}[
      scale=1.5,
      code/.style={ellipse, draw=black}
      ]
      \node[code] (call1) at (1, 0) {\texttt{calln}};
      \node[code] (call2) at (3.5, 0) {\texttt{callalpha}};

      \node[code] (tgt1) at (0, -1) {\texttt{fun\_1}};
      \node[code] (tgt2) at (1, -1.75) {\texttt{fun\_2}};
      \node[code] (tgta) at (2.5, -1) {\texttt{fun\_a}};
      \node[code] (tgtb) at (3.5, -1.75) {\texttt{fun\_b}};
      \node[code] (tgtc) at (4.5, -1) {\texttt{fun\_c}};

      \draw[->] (call1) [out=270, in=90] to 
        node[pos=0.1, right] {$\mathit{id}_S \gets 1$}
        node[pos=0.9, right] {$\mathit{id}_S = 1?$} (tgt1);
      \draw[dashed, ->] (call1) [out=270, in=90] to 
        node[pos=0.9, right] {$\mathit{id}_S = 1?$} (tgt2);
      \draw[dashed, ->] (call2) [out=270, in=90] to 
        node[pos=0.1, right] {$\mathit{id}_S \gets 2$}
        node[pos=0.9, right] {$\mathit{id}_S = 2?$} (tgta);
      \draw[->] (call2) [out=270, in=90] to 
        node[pos=0.9, right] {$\mathit{id}_S = 2?$} (tgtb);
      \draw[dashed, ->] (call2) [out=270, in=90] to 
        node[pos=0.9, right] {$\mathit{id}_S = 2?$} (tgtc);
    \end{tikzpicture}
    \caption{Fine-grained IBT~\cite{GaidisMSMAK23,KoruyehSKSA20}:
      Checks against a statically determined identifier either
      detect~\cite{GaidisMSMAK23} or prevent~\cite{KoruyehSKSA20} calls to functions
      outside a statically determined set.
      However, neither method protects against calls to an \emph{incorrect}
      function within the \emph{correct} set (dashed lines).\vspace{1em}}
    \label{fig:FineIBT}
  \end{subfigure}

  \begin{subfigure}{\linewidth}
    \begin{tikzpicture}[
      scale=1.5,
      code/.style={ellipse, draw=black}
      ]
      \node[code] (call1) at (1, 0) {\texttt{calln}};
      \node[code] (call2) at (3.5, 0) {\texttt{callalpha}};

      \node[code] (tgt1) at (0, -1) {\texttt{fun\_1}};
      \node[code] (tgt2) at (1, -1.75) {\texttt{fun\_2}};
      \node[code] (tgta) at (2.5, -1) {\texttt{fun\_a}};
      \node[code] (tgtb) at (3.5, -1.75) {\texttt{fun\_b}};
      \node[code] (tgtc) at (4.5, -1) {\texttt{fun\_c}};

      \draw[->] (call1) [out=270, in=90] to 
        node[pos=0.1, right] {$\callee \gets \mathtt{fun\_1}$}
        node[pos=0.9, right] {$\callee = \mathtt{fun\_1}?$} (tgt1);
      \draw[->] (call2) [out=270, in=90] to 
        node[pos=0.1, right] {$\callee \gets \mathtt{fun\_b}$}
        node[pos=0.9, right] {$\callee = \texttt{fun\_b}?$} (tgtb);
    \end{tikzpicture}
    \caption{\Triosecuris: Checks against the intended (dynamic) call target
      precisely detect when a misprediction has occurred.}
    \label{fig:Triosecuris}
  \end{subfigure}
  \caption{Granularity comparison of different IBT defenses against speculative executions attackers}
  \label{fig:IBT-comparison}
  \vspace{-1em}
\end{figure}

While originally designed to protect against \newreviews{sequential}
memory corruption attacks~\cite{CET, GaidisMSMAK23},
\newreviews[REQ4C]{hardware IBT protection} can also be leveraged to mitigate
speculative control-flow hijacking~\cite{GaidisMSMAK23, KoruyehSKSA20,
MosierNMT24, NarayanDMCJGVSS21}, \newreviews[REQ4C]{where software techniques are
infeasible. Spectre BTB attacks allow
jumps to any point in the code, making it infeasible to place software
mitigations at every possible speculative target.
Software techniques thus have to rely on checking the destination before
performing the jump or call. This is, however, impossible for Spectre BTB
attacks, where the misprediction is triggered by the call instruction itself and
cannot be prevented by any prior checks.

%
Hardware IBT, on the other hand, does not require any
instrumentation before the call instruction, and thus prevents speculative
calls to any location that does not contain the marker instruction.
However, hardware IBT by itself does not restrict which functions can be called
from where.}\removereviews{Some of these speculative
defenses~\mbox{\cite{MosierNMT24, NarayanDMCJGVSS21}}
are very coarse-grained, only marking locations that can be reached using
any indirect call, preventing jumps into the middle of a function,
but not restricting which functions can be called from where.}
As illustrated in \autoref{fig:IBT}, this allows an attacker to still
speculatively call any function even if it is not the intended target
(represented by the dotted and dashed lines).
Serberus~\cite{MosierNMT24} \removereviews{and Swivel~\cite{NarayanDMCJGVSS21}}\newreviews[REQ3]{uses
only such coarse-grained CFI protection, as they}
use additional mechanisms in the compiler to cope with this
imprecision\newreviews[REQ3]{ and obtain complete BTB protection (see \autoref{sec:related-work})}.

Other work~\cite{GaidisMSMAK23, KoruyehSKSA20} has applied ideas from
sequential control-flow hijacking prevention directly to mitigate Spectre BTB
by augmenting hardware CFI protections with labels, which are statically
computed at compile-time from the program's control flow graph. 
As illustrated in \autoref{fig:FineIBT}, all indirect calls and all indirectly callable functions are
annotated with a label, and an indirect call is only allowed to proceed if the
labels match, effectively limiting the targets to a limited set of
developer-intended functions. However, this 
does not yield complete protection against Spectre BTB, as it only
restricts mispredicted execution to a set of functions, which is statically
overapproximated, allowing clever attackers to still circumvent these defenses
in practice~\cite{WiebingTBG24}.

%
\begin{example}
  \autoref{lst:running-example-fineibt} shows FineIBT~\cite{GaidisMSMAK23}
  applied to the program in \autoref{lst:running-example}. Before the call, the
  identifier of a set of allowed functions is stored in a designated register
  $\mathit{id}_S$. All sequentially callable functions, here $\mathit{fun\_1}$
  and $\mathit{fun\_2}$, must share this identifier.
  The $\cctarget$ instruction (on Intel called ENDBR64 or ENDBR32)
  at the beginning of each function marks valid call targets.
  After the $\cctarget$, a conditional branch checks the
  identifier, branching to $\mathit{exit}$ in case of a mismatch.
  Note that the
  function $\mathit{fun\_a}$, which is not sequentially reachable from
  $l_\mathit{cont}$, checks against a different identifier.

  While this greatly reduces the possible attack surface, it does not offer
  complete protection. For example, the second attack described in
  \autoref{ex:btb-slh-bypass} is not prevented:
  Since $\mathit{fun\_1}$ and
  $\mathit{fun\_2}$ must share the same set identifier, a misprediction of one
  instead of the other cannot be detected. 
  \jb{we \emph{could} mention here that the inserted branches are vulnerable,
  but can be protected with USLH - but currently I've dropped PHT protection
entirely from this example}
\end{example}

\begin{listing}[FineIBT-hardened program]
  \label{lst:running-example-fineibt}
  We show only the basic blocks which are modified w.r.t.
  \autoref{lst:running-example}.
  
  \centering
  \begin{tikzpicture}
    \node[draw, rectangle, scale=0.8] (lcont) {
      \begin{minipage}{0.55\linewidth}
        \vspace{-1em}
        \begin{align*}
          l_\mathit{cont}: \quad& \highlight{\casgn{\mathit{id}_S}{1}}\\
          & \ccall \var{fun}\\
                                & \vspace{-.1em}\scalebox{0.8}{\large\vdots}
        \end{align*}
      \end{minipage}
    };
    \node[draw, rectangle, scale=0.8, right=0.2 of lcont] (fun2) {
      \begin{minipage}{0.6\linewidth}
        \vspace{-1em}
        \begin{align*}
          \mathit{fun\_2}: \quad&\highlight{\cctarget}\\
                                &\highlight{\cbranch{\mathit{id}_S \neq
          1}{\mathit{exit}}}\\
                                & \cload{\var{x}}{\var{base} +
        \var{arg1}}\\
                                & \cload{\var{y}}{\var{x}}\\
                                & \vspace{-.1em}\scalebox{0.8}{\large\vdots}
        \end{align*}
      \end{minipage}
    };
    \node[draw, rectangle, scale=0.8, below left=0.1 and 0.2 of fun2] (fun1) {
      \begin{minipage}{0.55\linewidth}
        \vspace{-1em}
        \begin{align*}
          \mathit{fun\_1}: \quad&\highlight{\cctarget}\\
                                &\highlight{\cbranch{\mathit{id}_S \neq
          1}{\mathit{exit}}} \\
                                & \vspace{-.1em}\scalebox{0.8}{\large\vdots}
        \end{align*}
      \end{minipage}
    };
    \node[draw, rectangle, scale=0.8, right=0.2 of fun1] (funa) {
      \begin{minipage}{0.6\linewidth}
        \vspace{-1em}
        \begin{align*}
          \mathit{fun\_a}: \quad&\highlight{\cctarget}\\
                                &\highlight{\cbranch{\mathit{id}_S \neq
          2}{\mathit{exit}}} \\
                                & \vspace{-.1em}\scalebox{0.8}{\large\vdots}
        \end{align*}
      \end{minipage}
    };
  \end{tikzpicture}
\end{listing}

\section{Key Ideas}
\label{sec:key-ideas}

\subsection{Precisely Detecting BTB Misprediction}
\label{sec:ki-btb}

In this work, we show that when using CET to defend against Spectre BTB, no
approximation is necessary. Instead of using statically computed labels, we make
direct use of the indirectly called function pointer, which we store
in a designated register immediately before the call,
as illustrated in \autoref{fig:Triosecuris}. All possible call
targets are also instrumented with a check comparing this designated register to
their actual location.
While this check is useless in the absence of speculation, it is able to detect
BTB mispredictions \emph{precisely}, as a mismatch
between the intended target at the ISA level (the value stored in our
designated register) and the misspeculated target location predicted by the BTB.
Thus, unlike FineIBT, this approach can be used to defend against all Spectre
BTB attacks.

As detection alone does not prevent Spectre attacks, we combine this approach
with Ultimate SLH (\autoref{sec:uslh}).
\removereviews{Our approach integrates nicely, as we can
directly update the same misspeculation flag, thus expanding the existing PHT
protection to also cover BTB misprediction.}
\newreviews[REQ4B1]{Since Ultimate SLH (as well as other SLH variants) mask
vulnerable data based on a \emph{misspeculation flag}, we can set
this flag when BTB misprediction is detected to apply the same protection. This
does not interfere with the existing detection mechanism for mispredicted
branches, so the resulting mitigation accurately detects both PHT and BTB
mispredictions. In combination with the detection mechanism for RSB
mispredictions (\autoref{sec:ki-rsb}), we obtain full protection against Spectre
PHT, BTB and RSB attacks, as well as combinations thereof.}

\begin{listing}[\Triosecuris-hardened program (calls)]
  \label{lst:running-example-specibt}
  We show only the basic blocks which are modified w.r.t. \autoref{lst:tracking-pht}
  and \autoref{lst:running-example-masking}.

  \centering
  \begin{tikzpicture}
    \node[draw, rectangle, scale=0.8] (lcont) {
      \begin{minipage}{0.4\linewidth}
        \vspace{-1em}
        \begin{align*}
          l_\mathit{cont}: \quad& \highlight{\casgn{\callee}{\var{fun}}}\\
          & \ccall \var{fun}\\
                                & \vspace{-.1em}\scalebox{0.8}{\large\vdots}
        \end{align*}
      \end{minipage}
    };
    \node[draw, rectangle, scale=0.8, right=0.15 of lcont] (fun2) {
      \begin{minipage}{0.72\linewidth}
        \vspace{-1em}
        \begin{align*}
          \mathit{fun\_2}: \quad&\highlight{\cctarget}\\
                                &\highlight{\casgn{\msf}{\ccond{\callee
                                \neq \efp{\mathit{fun\_2}}}{1}{\msf}}}\\
                                & \cload{\var{x}}{ \msf
          \operatorname{\texttt{?}} 0 \operatorname{\texttt{:}} \var{base} +
        \var{arg1}}\\
                                & \cload{\var{y}}{ \msf
          \operatorname{\texttt{?}} 0 \operatorname{\texttt{:}}\var{x}}\\
                                & \vspace{-.1em}\scalebox{0.8}{\large\vdots}
        \end{align*}
      \end{minipage}
    };
    \node[draw, rectangle, scale=0.8, below =0.1 of fun2] (fun1) {
      \begin{minipage}{0.72\linewidth}
        \vspace{-1em}
        \begin{align*}
          \mathit{fun\_1}: \quad&\highlight{\cctarget}\\
                                &\highlight{\casgn{\msf}{\ccond{\callee
                                \neq \efp{\mathit{fun\_1}}}{1}{\msf}}}\\
                                & \vspace{-.1em}\scalebox{0.8}{\large\vdots}
        \end{align*}
      \end{minipage}
    };
  \end{tikzpicture}
\end{listing}
\begin{example}
  \autoref{lst:running-example-specibt} shows \Triosecuris (including Ultimate SLH)
  applied to \autoref{lst:running-example}. In contrast to FineIBT
  (\autoref{lst:running-example-fineibt}), we make direct use of the called
  function pointer instead of a statically determined identifier. We also update
  the misspeculation flag directly, instead of introducing a branch.
  This version fully prevents both attack types shown in
  \autoref{ex:btb-slh-bypass}, and offers full protection against Spectre BTB
  attacks.
\end{example}

%
%
%

\newreviews[REQ2]{
\subsection{Precisely Detecting RSB Misprediction}
\label{sec:ki-rsb}

We use a similar approach to detect RSB misprediction.
Our approach relies on the fact
that, as described in \autoref{sec:bg-btb-rsb}, RSB
entries can only be created by call instructions\jb{cn}, thus,
the RSB can only predict targets which immediately follow a call.
%
%
\ch{+When the right flag is set that prevents the BTB from being used instead of the RSB}%
This restricts the set of possible target locations sufficiently that we can
also instrument each target location with a check. We use a similar mechanism as
in \autoref{sec:ki-btb}: Before a return, we read the top of the call stack
and store it in a designated register.
After each call instruction, we
place a check comparing this designated register to the actual location,
updating the misspeculation flag if a mismatch is detected.

\begin{listing}[\Triosecuris-hardened program (returns)]
  \label{lst:running-example-specibt-ret}
  We show only the basic blocks modified w.r.t. \autoref{lst:tracking-pht}
  and \autoref{lst:running-example-masking}.

  \centering
  \begin{tikzpicture}
    \node[draw, anchor=north west, rectangle, scale=0.8] (lcont) {
      \begin{minipage}{0.4\linewidth}
        \vspace{-1em}
        \begin{align*}
          l_\mathit{cont}: \quad & \casgn{\callee}{\var{fun}}\\
          & \ccall \var{fun}\\
          & \highlight{\casgn{\msf}{\ccond{\callee \neq \efp{l_\mathit{cont} +
          2}}{1}{\msf}}}\\
                                & \vspace{-.1em}\scalebox{0.8}{\large\vdots}
        \end{align*}
      \end{minipage}
    };
    \node[draw, rectangle, scale=0.8, right=0.2 of lcont.north east, anchor=north west] (fun2) {
      \begin{minipage}{0.4\linewidth}
        \vspace{-1em}
        \begin{align*}
          \mathit{fun\_2}: \quad& \vspace{-.1em}\scalebox{0.8}{\large\vdots}\\
                                & \highlight{\cpeek{\callee}}\\
                                & \cret
        \end{align*}
      \end{minipage}
    };
  \end{tikzpicture}
\end{listing}
\begin{example}
  \autoref{lst:running-example-specibt-ret} shows \Triosecuris (including Ultimate SLH)
  applied to \autoref{lst:running-example}. 
  The $\cpeek{\callee}$ instruction reads the head of the call stack and stores
  it in the designated register $\callee$, which is also used to detect call misspeculation.
\end{example}

}

\subsection{Undefined Values}
\label{sec:ki-undefined-values}


Correctly setting up the formal definitions and proving end-to-end relative
security for \Triosecuris was nontrivial, most interestingly because of undefined behavior.
We built \Triosecuris as an extension of Ultimate SLH, which has a realistic
implementation as one of the final passes of the LLVM compiler~\cite{ZhangBCSY23},
since this reduces the likelihood that compilation passes break security.
Working in such a low-level intermediate language is also convenient for \Triosecuris, since it
makes use of security features implemented in hardware (\IE the CET IBT instructions).
Yet, this language still features an abstract notion of labelled basic blocks
and, like in prior compiler formalizations~\cite{CompCertMM, ZakowskiBYZZZ21},
we make the semantics feature \emph{code pointer values}, which are distinct
from numerical values.
This allows us to still easily transform the code without changing the values of
the code pointers in the semantics of this intermediate language.

Since in mainstream compilers like LLVM all these values get compiled to just machine integers, the
distinction between code pointers and integers introduces undefined behavior in
the semantics, as the result of a comparison between a code pointer and a
numerical value cannot be determined before the program is compiled and laid out in memory.
In the presence of such undefined behavior, we can provide end-to-end security
guarantees only for safe inputs, which sequentially do not trigger undefined behavior.

Despite this seeming downside, defining the mitigation at this level is highly
advantageous compared to working at lower levels of abstraction or even directly
on machine code, as inserting instructions and blocks only requires tracking offsets locally within each block,
rather than across the whole program;\yh{In Fig. 6, we also track offsets for RSB protection.} further, working on machine code
would introduce the extra challenge of distinguishing whether constants are to
be interpreted as code pointers, and thus need to be adjusted, or as numeric
values, which must be left untouched. Finally, if a value were to be used as
both, it would not be possible to apply our mitigation while preserving the
sequential behavior of the program. We thus leave such behaviors undefined for
our language.

However, restricting our guarantees to inputs and executions free of such
undefined behaviors is not sufficient to be able to prove relative
security: 
Due to the Ultimate SLH masking applied to the address of loads, during
misspeculated execution all loads will read from the same masked address, in
our case zero ($0$). This applies independently of whether the load would
sequentially load a code pointer or a numerical value. Thus, this masking
would lead to undefined operations during speculative execution for any
operations that depend on previously loaded values.

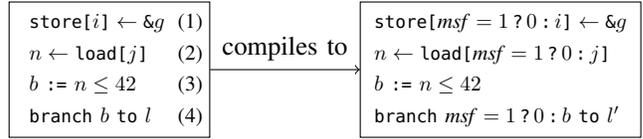
\begin{listing}[Ill-typed comparison induced by masking]
\renewcommand{\theHequation}{lst\thelisting.\theequation}
  \label{lst:ill-typed-masking}
%
  ~\\ 

  \vspace{-.5em}
  \centering
  \begin{tikzpicture}
    \node[draw, rectangle, scale=0.8] (src) {
      \begin{minipage}{0.35\linewidth}
        \vspace{-1em}
        \begin{alignat}{2}
          & \cstore{i}{\efp{g}}\\
          & \cload{n}{j}\\
          & \casgn{b}{n \leq 42}\label{eq:ex-ill-typed}\\
          & \cbranch{b}{l}
        \end{alignat}
      \end{minipage}
    };
    \node[draw, rectangle, scale=0.8] (tgt) at (5.2, 0) {
      \begin{minipage}{0.5\linewidth}
        \vspace{-1em}
        \begin{align*}
          &\cstore{\ccond{\msf = 1}{0}{i}}{\efp{g}}\\
          &\cload{n}{\ccond{\msf = 1}{0}{j}}\\
          &\casgn{b}{n \leq 42}\\
          & \cbranch{\ccond{\msf = 1}{0}{b}}{l'}
        \end{align*}
      \end{minipage}
    };
    \draw[->] (src) -- node[above]{compiles to} (tgt);
  \end{tikzpicture}
\end{listing}
\begin{example}
  \label{ex:ill-typed-masking}
  Consider the snippet on the left of \autoref{lst:ill-typed-masking} and its compilation on the right.
  Let's assume that sequentially, address $j$ contains a number, so that sequentially the comparison with the 42 in line
  \eqref{eq:ex-ill-typed} compares two numbers, which is a defined
  operation. However, during misspeculated execution of the transformed program,
  $\msf$ is set to 1 and thus now both the store and the load use address 0
  instead. Therefore, the load in line 2 reads the value $\efp{g}$, and the comparison
  in line 3 is now between a code pointer and a number, which is undefined.
\end{example}

If we were to treat these undefined operations as undefined behavior, our
speculative semantics would not describe the behavior of our mitigation. 
This would allow later compilation steps to refine it arbitrarily, potentially in
ways that leak secrets, which would break end-to-end relative security.
To prevent this, we ensure that our \Triosecuris mitigation does not introduce
undefined behavior.
In other words, we prove that our mitigation preserves safety: Any safe input,
\IE one that doesn't trigger undefined behavior sequentially in the source,
is also safe during speculative execution of the \Triosecuris-hardened program.
This then allows us to compose the relative security of \Triosecuris with the
relative security of our subsequent step (\autoref{fig:composition} and
\autoref{sec:translation-mc}), which only holds for safe inputs because it
refines undefined behavior.

In order for safety preservation to hold for \Triosecuris, we need to make examples
like the one on the right-side of \autoref{ex:ill-typed-masking} be defined in
the semantics of our intermediate language.
For this we introduce CompCert-style \emph{undefined values}~\cite{CompCertMM}
to represent values that cannot be determined at the current level of
abstraction, but that are refined to arbitrary concrete values during later
compilation steps, like the one in \autoref{sec:translation-mc}.
When an undefined value is created---\EG when speculatively comparing an integer
introduced by our masking with a code pointer, like on line 3 on the
right-side of \autoref{ex:ill-typed-masking}---the program proceeds normally.
Only if an undefined
value is used as the address of a load or store, or as the target of a call,
do we consider this to be undefined behavior.
However, undefined values which are the result of masking will not be used in
this way, as Ultimate SLH also masks all loads, stores, and calls,
which during speculative execution replaces any undefined values with constants.
\jb{Reviewer comment: This sentence is confusing, please revise}\ch{Started,
  please have a look.}

\begin{example}
  Recall the snippet in \autoref{ex:ill-typed-masking}.
  By allowing an undefined value for $b$ after line 3, we can let the execution proceed without undefined behavior:
  In the next line, where $b$ would be used, the expression is again masked with $\msf$, so that the result is independent of the (undefined) value of $b$.
\end{example}
Thus, \Triosecuris does not introduce undefined behavior:
\begin{lemma}[\Triosecuris preserves safety]
\label{lem:specibt-safety-preservation}
For any input for which the source program does not encounter undefined behavior
with respect to the sequential semantics of our intermediate language, the result of hardening this
program using \Triosecuris does not encounter undefined behavior with respect to the
speculative semantics of the language.
\end{lemma}

We used property-based testing in Rocq~\cite{Paraskevopoulou15} to discover and
fix issues with safety preservation before trying to prove it. In
particular, testing produced counterexamples like the one in
\autoref{lst:ill-typed-masking}, which we fixed by introducing undefined values.

\section{Defining \Triosecuris}
\label{sec:defining-specibt}
\subsection{Language Syntax and Sequential Semantics}
\label{sec:lang}
In this paper, we start from a simple low-level language with basic blocks
that we call \MiniMIR, because it is a simple model of the Machine IR (MIR) intermediate
language of LLVM where SLH is performed.
On top of this, \MiniMIR includes the CFI protections introduced by
\Triosecuris, which is a transformation from \MiniMIR to \MiniMIR (see
\autoref{fig:composition} in \autoref{sec:intro}).
The \MiniMIR syntax is shown in \autoref{fig:syntax}.
Programs consist of a list of basic blocks, where each block is a tuple of a list of
instructions (of type \texttt{inst}) and a boolean flag indicating whether the block
corresponds to the beginning of a function, in which case we refer to the
first instruction of this block as a \emph{function entry point}.

Expressions can, in addition to numeric constants, also contain
code-pointer constants. We write $\efp{(l, o)}$, where $(l, o)$ is a pair of a
label and an offset,
for the expression, and denote the corresponding value by $\vfp{(l, o)} \in \mathcal{CP}$.\ch{Didn't
  get this very last part. Does the difference in font matter, and anyway what
is this trying to say?}\jb{Is this better?}\ch{So we are using different fonts
to mean 2 different things?}\jb{Yes, this differentiates values from the
expressions that create them}
Also note the presence of a constant-time conditional expression ($\ccond{b}{e_1}{e_2}$),
which is essential for the SLH mitigation to work.
Our instruction set includes indirect calls
as well as a $\cctarget$ instruction that is used to mark valid call targets in
the speculative semantics, but otherwise behaves like a skip.
We assume that the source program does not use the $\cctarget$ instruction; it
will instead be inserted by our transformation (\autoref{sec:specibt-transformation}).
\newreviews[REQ2]{
  Further, our instruction set also includes a \texttt{peek} instruction,
  which stores the current topmost element of the return stack in the target
  register.\jb{Do we need to elaborate why this is its own instruction? mention
  explicitly that the return stack is not placed in memory? Do we need to
justify this as reading the shadow stack to preempt concerns about speculative
return address overwrites?}\ch{Okay to quickly elaborate on these.
  You seem to talk about the shadow stack in a couple of paras though.}
}


\begin{figure}
  \centering
  \[
  \begin{aligned}
    e \in\texttt{exp} ::=& ~ n\in\mathbb{N} && \text{number}\\
    |& ~ \efp{(l, o)} \in \mathcal{CP} && \text{code pointer}\\
    |& ~ \texttt{X} \in \mathcal{R}&&\text{register}\\
    |& ~ \op(e, \ldots, e)&&\text{arithmetic operator}\\
    |& ~ \ccond{b}{e}{e} &&\text{constant-time conditional}\\
    c\in\texttt{inst} ::=&~ \cskip &&\text{do nothing}\\
    |&~\casgn{\var{X}}{e}&&\text{assignment}\\
    |&~\cdiv{\var{X}}{e}{e}&&\text{division}\\
    |&~\cbranch{e}{l}&&\text{conditional branch}\\
    |&~\cjump{l}&&\text{unconditional jump}\\
    |&~\cload{\var{X}}{e}&&\text{load from address}\\
    |&~\cstore{e}{e}&&\text{store to address}\\
    |&~\ccall{e}&&\text{indirect call}\\
    |&~\cctarget&&\text{call-target marker}\\
    |&~\cpeek{\var{X}}&&\text{peek}\\
    |&~\cret&&\text{return}\\
  \end{aligned}
  \]
  \vspace{-1em}
  \caption{Language syntax}
  \label{fig:syntax}
  \vspace{-.5em}
\end{figure}

\ch{Shouldn't we call them registers instead of variables?
  AFAIR we already made this change at in the SpecCT SecF chapter already, right?
  Even if we agree this change may be too invasive to do at this point,
  but we should at least explain that our variables are similar to registers,
  in particular accessing them is not an observation.}
  \jb{done. Do we need to mention that we consider an unlimited number of
  registers?}\ch{Okay}

As explained in \autoref{sec:ki-undefined-values}, we consider not only natural
numbers and code pointer values $\vfp{(l, o)}$, but also an
undefined value $\UV$, which cannot be determined at this
level of abstraction (but will be refined to concrete values at lower levels, see
\autoref{sec:ki-undefined-values} and \autoref{sec:translation-mc}).
Accordingly, we define register assignments $\rho : \mathcal{R} \rightarrow
\mathbb{N} \cup \mathcal{CP} \cup \{\UV\}$
and memories $\mu : \mathbb{N} \rightarrow \mathbb{N} \cup \mathcal{CP} \cup
\{\UV\}$ to be able to store all three kinds of values, and define the
evaluation of expressions $\eval{e}{\rho}$ as described in
\autoref{fig:eval-exp}.
We write $\subst{\var{X}}{v}{\rho}$ to denote updating the value of $\var{X}$ to $v$ in
$\rho$, and similarly $\subst{i}{v}{\mu}$ to denote updating the value at index
$i$ to $v$. Memory lookups are written $\mu(i)$.

\begin{figure}
\centering
  \begin{align*}
  \op &::= \kw{+} ~\barB~ \kw{-} ~\barB~ \kw{*} ~\barB~ \kw{=} ~\barB~
          \kw{\leq} ~\barB~ \kw{\land} ~\barB~ \kw{\rightarrow} 
  \end{align*}
  \vspace{-1em}
  \begin{align*}
  v_1\, \op \, v_2 &=
    \begin{cases}
      (l_1, o_1) = (l_2, o_2) & \text{if } (v_1 = \vfp{(l_1, o_1)}, \\
      & \phantom{\text{if } (} v_2 = \vfp{(l_2, o_2)})\\
      & \op \text{ is } \kw{=} \\
      \textit{standard} & \text{if } v_1, v_2 \in \mathbb{N}\\
      \UV & \text{otherwise}
    \end{cases} \\
  \eval{n}{\rho} &= n \\
  \eval{\efp{(l, o)}}{\rho} &= \vfp{(l, o)} \\
  \eval{x}{\rho} &= \rho(x) \\
  \eval{e_1 \, \op \, e_2}{\rho} &= \eval{e_1}{\rho} \; \op \; \eval{e_2}{\rho} \\
  \eval{\ccond{b}{e_1}{e_2}}{\rho} &= \ite{\eval{b}{\rho} \neq 0}{\eval{e_1}{\rho}}{\eval{e_2}{\rho}}
  \end{align*}
\caption{Evaluation of expressions. We assume the standard interpretation of
binary operators on integers. Note that the conditional expression
$\ccond{b}{e_1}{e_2}$ can still produce a defined value if only the unused
expression is undefined.}
\label{fig:eval-exp}
\end{figure}

Beyond the points above, the small-step operational semantics to \MiniMIR is
standard.
It can be obtained from \autoref{fig:spec-semantics-specibt} by ignoring the
\highlight{highlighted parts}, which will be added by the speculative
semantics\ifappendix~(and is also given verbatim in \autoref{fig:spec-semantics-minimc})\fi.
Program states are either a tuple of a program
counter $\mathit{pc}$, a register assignment $\rho$, a memory $\mu$ and a stack
$\mathit{stk}$ of return locations, or the special state $\termstate$
representing successful termination.
Program counter values consist of a basic block label and an offset within the block.
\newreviews[REQ2]{
The stack $\mathit{stk}$ abstractly represents a CET-protected return stack. We
do not model the combination of a normal stack in memory with an additional
protected shadow stack, instead, our language only relies on the stack abstractly represented
within program states, which is modified only by \texttt{call} and \texttt{ret}
instructions, and can be read from using \texttt{peek}.
}

We write $\prog \vdash \state{\mathit{pc}, \rho, \mu, \mathit{stk}} \seqstep{o}
\state{\mathit{pc}', \rho', \mu', \mathit{stk}'}$ for a step optionally
producing an \emph{observation} $o$, or $\prog \hspace{-2pt}\vdash\hspace{-2pt} \state{\mathit{pc}, \rho,
\mu, \mathit{stk}} \seqstep{o} \cdot$ if we do not care about the state after
the step. As the program remains fixed throughout the execution, we may omit it
and the turnstile from the relation, although it is used in the premises.

Observations $o \in \mathit{Option}(\mathit{Obs})$ represent the information
that can potentially be gained by an attacker via side channels.
We use a standard\jb{Is this still standard?}\ch{Why not?} model for sequential attackers,
which observe any accessed memory addresses ($\oload{a}$ and $\ostore{a}$ resp.
for read and write operations at address $a$), the control flow of a program
($\obranch{b}$ for a branch whose condition evaluates to $b$ and $\ocall{l}$ for
a call to label $l$), \newtext{as well as the operands of the division
instruction ($\odiv{n_1}{n_2}$).}
For steps that do not produce an observation, we write $\bullet$.
%
Expression evaluation does not produce observable events, since as
mentioned above we assume the conditional expression
$\ccond{\mathit{b}}{e_1}{e_2}$ is internally implemented using branchless and
unpredicted logic (such as conditional moves~\cite{Carruth18}).

Multi-step execution is the reflexive and transitive closure of the step
relation, written $\seqmulti{\Obss}$, where $\Obss \in
\mathit{List}(\mathit{Obs})$ is a trace of events.
Silent steps leave no mark in the trace, in other words,
$\Obss$ contains only observations that are not $\bullet$.

\subsection{Speculative Semantics}

\label{sec:spec-sem}

\begin{figure}
  \centering
  \resizebox{\columnwidth}{!}{%
  \(
  \begin{array}{c}
  \inferrule*[left=\textsc{Spec\_Skip}]
    { \fetch{\prog}{\pc}{\cskip} }
    { \runningstate{\scfg{\cfg{\pc}{\rho}{\mu}{\stk}}{\highlight{\mfalse}}{\highlight{\ms}}}
      \specstep{\bullet}{\highlight{\bullet}}
      \runningstate{\scfg{\cfg{\pc+1}{\rho}{\mu}{\stk}}{\highlight{\mfalse}}{\highlight{\ms}}}
    }
  \\[1.2ex]
  \inferrule*[left=\textsc{Spec\_Asgn}]
    { \fetch{\prog}{\pc}{\casgn{x}{e}} \quad \eval{e}{\rho} = v }
    { \runningstate{\scfg{\cfg{\pc}{\rho}{\mu}{\stk}}{\highlight{\mfalse}}{\highlight{\ms}}}
      \specstep{\bullet}{\highlight{\bullet}}
      \runningstate{\scfg{\cfg{\pc+1}{\subst{x}{v}{\rho}}{\mu}{\stk}}{\highlight{\mfalse}}{\highlight{\ms}}}
    }
  \\[1.2ex]
  \inferrule*[left=\textsc{Spec\_Div}]
    { \fetch{\prog}{\pc}{\cdiv{x}{e_1}{e_2}} \quad \eval{e_1}{\rho} = n_1 \quad \eval{e_2}{\rho} = n_2 \quad
    n = {\begin{cases} n_1/n_2 & \text{if } n_2 \neq 0 \\ \UV & \text{if } n_2 = 0 \end{cases}}
    }
    { \runningstate{\scfg{\cfg{\pc}{\rho}{\mu}{\stk}}{\highlight{\mfalse}}{\highlight{\ms}}}
      \specstep{\odiv{n_1}{n_2}}{\highlight{\bullet}}
      \runningstate{\scfg{\cfg{\pc+1}{\subst{x}{n}{\rho}}{\mu}{\stk}}{\highlight{\mfalse}}{\highlight{\ms}}}
    }
  \\[1.2ex]
  \inferrule*[left=\textsc{Spec\_Branch}]
  { \substack{
    \fetch{\prog}{\pc}{\cbranch{e}{l}} \quad \eval{e}{\rho} = n \quad b = (n \neq 0) \\
     \pc' = \ite{\highlight{b'}}{(l,0)}{\pc+1} \quad \highlight{\ms' = \ms \lor (b \neq b')}
   } }
  { \runningstate{\scfg{\cfg{\pc}{\rho}{\mu}{\stk}}{\highlight{\mfalse}}{\highlight{\ms}}}
    \specstep{\obranch{b}}{\highlight{\dbranch{b'}}}
    \runningstate{\scfg{\cfg{\pc'}{\rho}{\mu}{\stk}}{\highlight{\mfalse}}{\highlight{\ms'}}}
  }
  \\[1.2ex]
  \inferrule*[left=\textsc{Spec\_Jump}]
  { \fetch{\prog}{\pc}{\cjump{l}} }
  { \runningstate{\scfg{\cfg{\pc}{\rho}{\mu}{\stk}}{\highlight{\mfalse}}{\highlight{\ms}}}
    \specstep{\bullet}{\highlight{\bullet}}
    \runningstate{\scfg{\cfg{(l,0)}{\rho}{\mu}{\stk}}{\highlight{\mfalse}}{\highlight{\ms}}} }
  \\[1.2ex]
  \inferrule*[left=\textsc{Spec\_Load}]
  { \fetch{\prog}{\pc}{\cload{x}{e}} \quad
    \eval{e}{\rho} = n \quad \mu [n] = v
  }
  { \runningstate{\scfg{\cfg{\pc}{\rho}{\mu}{\stk}}{\highlight{\mfalse}}{\highlight{\ms}}}
    \specstep{\oload{n}}{\highlight{\bullet}}
    \runningstate{\scfg{\cfg{\pc+1}{\subst{x}{v}{\rho}}{\mu}{\stk}}{\highlight{\mfalse}}{\highlight{\ms}}}
  }
  \\[1.2ex]
  \inferrule*[left=\textsc{Spec\_Store}]
  {
    \fetch{\prog}{\pc}{\cstore{e}{e'}} \quad \eval{e}{\rho} = n \quad \eval{e'}{\rho} = v
  }
  {
    \runningstate{\scfg{\cfg{\pc}{\rho}{\mu}{\stk}}{\highlight{\mfalse}}{\highlight{\ms}}}
    \specstep{\ostore{n}}{\highlight{\bullet}}
    \runningstate{\scfg{\cfg{\pc+1}{\rho}{\subst{n}{v}{\mu}}{\stk}}{\highlight{\mfalse}}{\highlight{\ms}}}
  }
  \\[1.2ex]
  \inferrule*[left=\textsc{Spec\_Call}]
  {
    \fetch{\prog}{\pc}{\ccall{e}} \quad \eval{e}{\rho} = \vfp{(l, o)} \quad \highlight{\ms' = \ms \lor \pc' \neq (l, o)}
  }
  {
    \runningstate{\scfg{\cfg{\pc}{\rho}{\mu}{\stk}}{\highlight{\mfalse}}{\highlight{\ms}}}
    \specstep{\ocall{(l, o)}}{\highlight{\dcall{\pc'}}}
    \runningstate{\scfg{\cfg{\pc'}{\rho}{\mu}{(\pc+1)::\stk}}{\highlight{\mtrue}}{\highlight{\ms'}}}
  }
  \\[1.2ex]
  \inferrule*[left=\textsc{Spec\_Peek}]
  { \fetch{\prog}{\pc}{\cpeek{x}} \quad
    v = {\begin{cases} \pc' & \text{if } \stk = \pc' :: \stk' \\ \UV & \text{if } \stk = \epsilon \end{cases}}
  }
  { \runningstate{\scfg{\cfg{\pc}{\rho}{\mu}{\stk}}{\highlight{\mfalse}}{\highlight{\ms}}}
    \specstep{\bullet}{\highlight{\bullet}}
    \runningstate{\scfg{\cfg{\pc+1}{\subst{x}{v}{\rho}}{\mu}{\stk}}{\highlight{\mfalse}}{\highlight{\ms}}}
  }
  \\[1.2ex]
  \inferrule*[left=\textsc{Spec\_Ret}]
  {
    \fetch{\prog}{\pc}{\cret} \quad \highlight{\wfret{\pc''}} \quad \highlight{\ms' = \ms \lor \pc' \neq \pc''}
  }
  {
    \runningstate{\scfg{\cfg{\pc}{\rho}{\mu}{\pc' :: \stk}}{\highlight{\mfalse}}{\highlight{\ms}}}
    \specstep{\bullet}{\highlight{\dret{\pc''}}}
    \runningstate{\scfg{\cfg{\pc''}{\rho}{\mu}{\stk}}{\highlight{\mfalse}}{\highlight{\ms'}}}
  }
  \\[1.2ex]
  \inferrule*[left=\textsc{Spec\_Term}]
  {
    \fetch{\prog}{\pc}{\cret}
  }
  {
    \runningstate{\scfg{\cfg{\pc}{\rho}{\mu}{\emptylist}}{\highlight{\mfalse}}{\highlight{\ms}}}
    \specstep{\bullet}{\highlight{\bullet}}
    \termstate
  }
  \\[1.2ex]
  \highlight{\inferrule*[left=\textsc{Spec\_CTarget}]
  {
  \fetch{\prog}{\pc}{\cctarget}
  }
  {
    \runningstate{\scfg{\cfg{\pc}{\rho}{\mu}{\stk}}{\ct}{\ms}}
    \specstep{\bullet}{\bullet}
    \runningstate{\scfg{\cfg{\pc+1}{\rho}{\mu}{\stk}}{\mfalse}{\ms}}
  }}
  \\[1.2ex]
  \highlight{\inferrule*[left=\textsc{Spec\_Fault}]
  { \prog \llbracket \pc \rrbracket \neq \cctarget }
  {
    \runningstate{\scfg{\cfg{\pc}{\rho}{\mu}{\stk}}{\mtrue}{\ms}}
    \specstep{\bullet}{\bullet}
    \faultstate
  }}
  \end{array}\)
  }
  \caption{Speculative semantics for \MiniMIR}
  \label{fig:spec-semantics-specibt}
  \vspace{-1em}
\end{figure}

The sequential semantics and its observations only allow us to
reason about leakage that occurs due to the intended control flow
and memory accesses of programs. Spectre attacks, however,
rely on observations produced during \emph{transient} execution, when the
processor \emph{speculatively} begins to execute the instructions along an
incorrect path before the instruction which determines the correct path has
fully resolved. Furthermore, speculative attackers are able to actively influence
these transient executions, by training the hardware prediction mechanisms to
predict an attacker-chosen path.

\newreviews[REQ4B2]{
  In order to model an attacker that can mount PHT, BTB and RSB attacks, we give
  the attacker direct control over all three prediction mechanisms using
  \emph{directives}. 
  As shown by \citet{BartheCGKLOPRS21}, by quantifying over all sequences of
  directives, this securely overapproximates attacker capabilities, even without
  modeling a rollback mechanism and hardware details such as
  finite speculation windows.\jb{We can also cite my Master's project here, but only for the not-anonymized
  version}
  Because of its simplicity, this directives-based model is quite popular
  in cryptographic settings~\cite{BartheCGKLOPRS21, OlmosBBGL25, OlmosBCGLOSYZ25,
  ShivakumarBBCCGOSSY23, ShivakumarBGLOPST23}
  and beyond~\cite{ZhangBCSY23, BaumannBDHH25}.
  Moreover, \citet{GuarnieriKRV21} have shown that many other low-level details such as
  out-of-order execution can safely be ignored, as they do not cause any further
  leakage.\jb{This could probably be phrased a lot better}
}

\remove{
Our mitigation targets hardware IBT protection, which requires specific markers
($\cctarget$) to be inserted in all valid callable locations. In order to ensure
that our mitigation places these markers correctly, we need to model both an
attacker's ability to misdirect execution to arbitrary locations even within
functions, as well as a hardware mechanism that halts execution if no
$\cctarget$ marker is found.
}
\newtext{
  Further, our speculative semantics also needs to model hardware IBT
  protection, requiring that specific markers ($\cctarget$) are inserted at all
  variable locations, and halting execution if no such marker is found after a
  call. We thus use%
}%
\remove{Therefore, to model attacks in this setting, we use} the speculative semantics with
CFI protection shown in \autoref{fig:spec-semantics-specibt}. This semantics is
a strict extension of the sequential semantics\ifappendix\space(\autoref{sec:appendix-seq-semantics})\fi, which can be recovered by
removing the \highlight{highlighted additions}.
The speculative small-step semantics $\prog \vdash \state{\mathit{pc}, \rho,
\mu, \mathit{stk}, \mathit{ct}, \mathit{ms}} \specstep{o}{d}
\state{\mathit{pc}', \rho', \mu', \mathit{stk}', \mathit{ct}', \mathit{ms}'}$ 
 differs from the sequential semantics in the following ways:
\begin{itemize}[leftmargin=*,nosep,label=$\blacktriangleright$]
  \item An optional directive $d \in \mathit{Option}(\mathit{Dir})$ overapproximates
    the attacker's ability to cause execution to go down a specific
    speculative path, either by forcing a conditional branch
    with $\dbranch{b}$, by choosing the transient target of an indirect call
    with $\dcall{\pc}$ (with an arbitrary program counter value, i.e. label and
    offset), \newreviews[REQ2]{or by choosing the transient target of a return with $\dret{\pc}$
    (where the predicate $\wfret{\pc}$ restricts targets to valid return addresses,
    \IE locations immediately following a call instruction;
    this restriction is justified by the {\tt RRSBA\_DIS\_U} flag)}.
    For other instructions, no directive is needed, so we write $\bullet$
    instead.
  \item We introduce a new special state $\faultstate$, to clearly distinguish
    the faulty termination as the result of a CFI violation from both normal
    termination ($\termstate$) and from undefined behavior (stuck states).
  \item Speculative program states include a flag $\mathit{\ct}$, which tracks
    whether a $\cctarget$ instruction is expected as the next instruction. This
    flag is set to true during a call (rule \textsc{Spec\_Call}) and reset by
    the $\cctarget$ instruction (rule \textsc{Spec\_CTarget}).
    While this flag is set, any other instruction will trigger a fault (rule
    \textsc{Spec\_Fault}).
  \item We extend speculative program states with a misspeculation flag $\ms$,
    set by \textsc{Spec\_Branch}, \textsc{Spec\_Call}, and \textsc{Spec\_Ret} to track whether
    execution has diverged from the sequential path. It is not needed for the
    semantics, but helps distinguish sequential and speculative execution in our
    proofs.
\end{itemize}
As with the sequential semantics, we often elide the program $\prog$ and the
turnstile from the relation, as it is unchanged throughout the execution.
We write $\specmulti{\Obss}{\Dirs}$ for multi-step execution, where
$\Dirs \in \mathit{List}(\mathit{Dir})$ is a sequence of attacker
directives.
\removereviews[REQ4B2]{
Because our definition of relative security quantifies over all execution paths,
our speculative semantics does not need to model hardware
details~
\jb{I can't infer the point of citing \cite{BrotzmanZKT21} - are we citing
examples of works that do model hardware details? (they seem to model caches).
\cite{GuancialeBD20} use an OoO semantics}
such a finite speculation window~\cite{OleksenkoFKS22},
not even rollbacks~\cite{BartheCGKLOPRS21}.
}


\subsection{\Triosecuris Transformation}
\label{sec:specibt-transformation}

\begin{figure}[t]

\begin{subfigure}{\columnwidth}
\scriptsize

\[
\begin{aligned}
\tr{i} \triangleq (\bar{\imath},\Delta)
\qquad
&\text{where }\bar{\imath}\text{ is emitted code (a basic block fragment) and } \\
&\Delta\text{ is a list of added basic blocks.}
\\
&\bar{\imath} : \mathsf{Blk} \triangleq \mathsf{list}\ \mathsf{inst}
\qquad
\Delta : \mathsf{Prog} \triangleq \mathsf{list}\ \mathsf{Blk}.
\end{aligned}
\]

\vspace{-1em}

\begin{align*}
\tr{\cdiv{x}{e_1}{e_2}} &= ([\cdiv{x}{(\ccond{\msf}{0}{e_1})}{(\ccond{\msf}{0}{e_2})}],[]) \\
\tr{\cload{x}{e}} &= ([\cload{x}{\ccond{\msf}{0}{e}}],[]) \\
\tr{\cstore{e}{e'}} &= ([\cstore{\ccond{\msf}{0}{e}}{e'}],[]) \\
\tr{\cbranch{e}{l}} &= \begin{array}[t]{@{}l@{}}
                       ([\cbranch{e'}{l^\star}; \casgn{\msf}{\ccond{e'}{1}{\msf}}],\Delta^\star) \\
                        \text{where } e' = \ccond{\msf}{0}{e}, l^\star \text{ is a fresh label.} \\
                        \Delta^\star = [\casgn{\msf}{\ccond{\neg e'}{1}{\msf}}; \cjump{l}]
                       \end{array} \\
\tr{\ccall{e}} &= \begin{array}[t]{@{}l@{}}
                  ([\casgn{\callee}{e'}; \ccall{e'}; \casgn{\msf}{\ccond{e''}{\msf}{1}}],[]) \\
                  \text{where } e' = \ccond{\msf}{\texttt{\&}(0, 0)}{e} \\
                  \text{and } e'' = (\callee = \tpc), \tpc \text{ is the pc of the assignment in } \tr{p}
                  \end{array} \\
\tr{\cret} &= ([\cpeek{\callee}; \cret],[]) \\
\end{align*}
\vspace{-3.7em}
\caption{Translation of individual instructions}
\label{fig:specibt-transformation-inst}
\end{subfigure}
\vspace{0.2em}

\begin{subfigure}{\columnwidth}
  \scriptsize
\[
\begin{aligned}
\tr{\mathit{blk}}
\triangleq\;&
\text{let }(\bar{\imath}_1,\Delta_1),\dots,(\bar{\imath}_k,\Delta_k)
\text{ be translations of }i_j\in\mathit{blk}
\\
&\text{in }\Big(\bar{\imath}_1\concat \cdots\concat \bar{\imath}_k,\ 
\Delta_1\concat \cdots\concat \Delta_k\Big).
\end{aligned}
\]

\vspace{0.2em}

\[
\mathit{pre}_l \triangleq [\ \cctarget;\ \msf := (\ccond{(\callee=\efp{(l, 0)})}{\msf}{1})\ ].
\]

\vspace{0.1em}

\[
\begin{aligned}
\tr{(l,\mathit{blk})}
\triangleq\;&
\text{let }(\mathit{blk}',\Delta)=\tr{\mathit{blk}}\text{ in}
\\
&
\begin{cases}
(\mathit{pre}_l \concat \mathit{blk}',\ \Delta)
& \text{if } l \text{ is a procedure entry basic block.}
\\
(\mathit{blk}',\ \Delta)
& \text{otherwise.}
\end{cases}
\end{aligned}
\]
\vspace{-1.5em}
\caption{Translation of basic blocks}
\label{fig:specibt-transformation-blk}
\end{subfigure}
\vspace{0.5em}
\begin{subfigure}{\columnwidth}
  \scriptsize
\[
\begin{aligned}
\tr{p}
\triangleq\;& \text{let }((p',\Delta_{\mathrm{all}}),\_) \triangleq
\mathsf{mapAccum}\Big(\lambda (k,\mathit{blk}).\ \tr{(k,\mathit{blk})},\ \lengthof{\prog},\ \prog\Big)
\text{ in}
\\
& \prog' \concat \Delta_{\mathrm{all}}.
\end{aligned}
\]
\vspace{-1.5em}
\caption{Translation of full programs}
\label{fig:specibt-transformation-prog}
\end{subfigure}

\caption{\Triosecuris transformation}
\label{fig:specibt-transformation}
\vspace{-.5em}
\end{figure}

Our \Triosecuris transformation makes use of two reserved registers: $\msf$, which
contains the misspeculation flag, and $\callee$\jb{consider renaming}, which tracks the intended
target of indirect calls \newreviews[REQ2]{and returns}.
The transformation inserts instructions to maintain these flags and use them to
mask leaking operands. Thus, at a high level, it needs to accomplish the following
tasks:
\begin{itemize}[leftmargin=*,nosep,label=$\blacktriangleright$]
  \item For every conditional branch, in both branches, the misspeculation flag
    must be updated with the condition used in the branch. While this is
    trivial if the branch is not taken (as we can simply insert the
    corresponding instruction), it is more difficult in case the branch is
    taken: The next instruction that will be executed is not at the current
    location, but at the target label. However, we cannot insert the instruction
    there either, as multiple branch or jump instructions could target the same
    label, and the instruction we insert to update the misspeculation flag must
    not interfere with those.
    Instead, we perform \emph{edge-splitting} by creating a new basic block which
    contains only the update of the misspeculation flag and a jump to the
    original target.
  \item Before each call instruction, $\callee$ must be set to the intended
    target. Unlike the previous case, this is simple, as all
    function entry points are known in advance and are handled in the exact same
    way.
  \item Basic blocks corresponding to function entry points must be prefixed with
    a $\cctarget$ instruction, followed by a check to ensure the $\callee$
    register matches the current basic block, and updating the misspeculation flag
    accordingly.
  \item \newreviews[REQ2]{Before every return instruction, $\callee$ must be set to the correct
    target, which is obtained by the $\cpeek{\callee}$ instruction.}
  \item \newreviews[REQ2]{After every call instruction (\IE every instruction that can be reached
    by a return), a $\callee$ check must be placed to update the misspeculation
    flag in case the location does not match the location read from the return
  stack.}
  \item Like in Ultimate SLH,
    all addresses of stores and loads, \newtext{operands of division
    instructions,} as well as the conditions of branches
    and the target expressions of indirect calls, must be masked based on the
    misspeculation flag. This only requires modifying the expression by wrapping
    it in a constant-time conditional $\ccond{\msf}{0}{e}$ (or
    $\ccond{\msf}{\efp{0}}{e}$ for indirect calls).
\end{itemize}

The definition in \autoref{fig:specibt-transformation} decomposes into three parts.
The first is the translation of individual instructions
(\autoref{fig:specibt-transformation-inst}), which returns a sequence of
instructions as well as a list of new basic blocks.
\autoref{fig:specibt-transformation-blk} handles blocks by concatenating the
translations of their individual instructions, and adding the prelude consisting
of a $\cctarget$ instruction and the $\msf$ update to blocks that are
marked as function entry points. Finally, at the program level, we translate
each block individually and then append all newly created blocks
(\autoref{fig:specibt-transformation-prog}).

\section{Formal Results for \Triosecuris}
\label{sec:formal-results}

In this section, we prove that \Triosecuris (\autoref{fig:specibt-transformation}) guarantees
relative security (\autoref{sec:bg-rs}) for all programs. \newtext{Our proof
  follows the same high-level structure as the relative security proof of FSLH~\cite{BaumannBDHH25},
but is more involved due to the presence of indirect calls and returns. }

\subsection{Ideal Semantics}

\begin{figure}
  \centering
  \resizebox{\columnwidth}{!}{%
  \(
  \begin{array}{c}
  \inferrule*[left=\textsc{Ideal\_Div}]
    { \substack{
      \fetch{\prog}{\pc}{\cdiv{x}{e_1}{e_2}} \quad \highlight{\ite{\ms}{0}{\eval{e_1}{\rho}} = n_1} \\
      \highlight{\ite{\ms}{0}{\eval{e_2}{\rho}} = n_2} \quad
      n = {\begin{cases} n_1/n_2 & \text{if } n_2 \neq 0 \\ \UV & \text{if } n_2 = 0 \end{cases}}
    }
    }
    { \runningstate{\icfg{\cfg{\pc}{\rho}{\mu}{\stk}}{\ms}}
      \specstep{\odiv{n_1}{n_2}}{\bullet}
      \runningstate{\icfg{\cfg{\pc+1}{\subst{x}{n}{\rho}}{\mu}{\stk}}{\ms}}
  }
  \\[1.2ex]
  \inferrule*[left=\textsc{Ideal\_Branch}]
  { \substack{
    \fetch{\prog}{\pc}{\cbranch{e}{l}} \quad \highlight{\ite{\ms}{0}{\eval{e}{\rho}} = n} \quad b = (n \neq 0) \\
     \pc' = \ite{b'}{(l,0)}{\pc+1} \quad \ms' = \ms \lor (b \neq b')
   } }
  { \runningstate{\icfg{\cfg{\pc}{\rho}{\mu}{\stk}}{\ms}}
    \idealstep{\obranch{b}}{\dbranch{b'}}
    \runningstate{\icfg{\cfg{\pc'}{\rho}{\mu}{\stk}}{\ms'}}
  }
  \\[1.2ex]
  \inferrule*[left=\textsc{Ideal\_Load}]
  { \fetch{\prog}{\pc}{\cload{x}{e}} \quad \highlight{\eval{\ite{ms}{0}{e}}{\rho} = n} \quad \mu [n] = v
  }
  { \runningstate{\icfg{\cfg{\pc}{\rho}{\mu}{\stk}}{\ms}}
    \idealstep{\oload{n}}{\bullet}
    \runningstate{\icfg{\cfg{\pc+1}{\subst{x}{v}{\rho}}{\mu}{\stk}}{\ms}}
  }
  \\[1.2ex]
  \inferrule*[left=\textsc{Ideal\_Store}]
  {
    \fetch{\prog}{\pc}{\cstore{e}{e'}} \quad \highlight{\eval{\ite{ms}{0}{e}}{\rho} = n} \quad \eval{e'}{\rho} = v
  }
  {
    \runningstate{\icfg{\cfg{\pc}{\rho}{\mu}{\stk}}{\ms}}
    \idealstep{\ostore{n}}{\bullet}
    \runningstate{\icfg{\cfg{\pc+1}{\rho}{\subst{n}{v}{\mu}}{\stk}}{\ms}}
  }
  \\[1.2ex]
  \inferrule*[left=\textsc{Ideal\_Call}]
  {
    \fetch{\prog}{\pc}{\ccall{e}} \quad \highlight{\ite{\ms}{\vfp{(0, 0)}}{\eval{e}{\rho}} = \vfp{(l, o)}} \quad \ms' = \ms \lor (l', o') \neq (l, o) \\
    \highlight{\lengthof{\prog} > l' \quad \isct{\prog[l']}} \quad \highlight{o' = 0}
  }
  {
    \runningstate{\icfg{\cfg{\pc}{\rho}{\mu}{\stk}}{\ms}}
    \idealstep{\ocall{(l, o)}}{\dcall{(l', o')}}
    \runningstate{\icfg{\cfg{\pc'}{\rho}{\mu}{(\pc+1)::\stk}}{\ms'}}
  }
  \\[1.2ex]
  \inferrule*[left=\textsc{Ideal\_Call\_Fault}]
  {
    \fetch{\prog}{\pc}{\ccall{e}} \quad \highlight{\ite{\ms}{\vfp{(0, 0)}}{\eval{e}{\rho}} = \vfp{(l, o)}} \\
    \highlight{\lengthof{\prog} \leq l' \lor \neg\isct{\prog[l']} \lor (o' \neq 0)}
  }
  {
    \runningstate{\icfg{\cfg{\pc}{\rho}{\mu}{\stk}}{\ms}}
    \idealstep{\ocall{(l, o)}}{\dcall{(l', o')}}
    \faultstate
  }
  \\[1.2ex]
  \end{array}\)
  }

  \caption{Ideal semantics for \MiniMIR (selected rules)\ch{What's the meaning
  of the color highlight?}\jb{see text}}
  \label{fig:ideal-semantics-specibt}
\end{figure}
In order to simplify reasoning about the security of our \Triosecuris semantics, we
follow the approach of \citet{ShivakumarBBCCGOSSY23}
and introduce an auxiliary semantics to represent the idealized behavior that
we wish to enforce with the \Triosecuris transformation.
This \emph{ideal semantics} therefore models speculative execution, but with all
masking enforced directly in the semantics; similarly, this semantics does not
model hardware IBT, but directly enforces that attacker attempts to force control
flow to an invalid target lead to a fault.
%

Again, we define this semantics as a small-step relation, written $\prog \vdash
\state{\pc, \rho, \mu, \stk, \ms} \idealstep{\Obss}{\Dirs} \, \st'$ (resp.
$\prog \vdash \st \idealmulti{\Obss}{\Dirs} \, \st'$ for multi-step execution).
As with the sequential and speculative semantics, we often omit the program
$\prog$ and the turnstile from the equation.
The ideal semantics shown in \autoref{fig:ideal-semantics-specibt} is derived from the speculative semantics shown
in \autoref{fig:spec-semantics-specibt}, with the following
\highlight{highlighted} notable changes \ifappendix\space(The full set of rules can be found in \autoref{sec:appendix-ideal-semantics})\fi:
\begin{itemize}[leftmargin=*,nosep,label=$\blacktriangleright$]
  \item \textsc{Ideal\_Load} and \textsc{Ideal\_Store} conditionally mask the
    address using the misspeculation flag $\mathit{ms}$ in order to ensure that
    there is no leakage via addresses. Recall that this flag
    is directly tracked by the speculative semantics, and therefore always
    accurately reflects whether execution is currently on an incorrect path.
    Similarly, \textsc{Ideal\_Branch} masks the condition, ensuring that there
    is no leakage via control flow.
  \item The rule \textsc{Ideal\_Call} not only masks the target expression
    during speculative execution, but also enforces that the target location
    specified by the attacker directive is a valid call target
    \newtext{--- \IE the first instruction ($o' = 0$) of procedure entry basic block ($\isct{\prog[l']}$) ---}
    \remove{(}and thus should not trigger a fault\remove{)}.
  \item In case the attacker directive specifies an invalid call target, rule
    \textsc{Ideal\_Call\_Fault} specifies that the call will trigger a fault.
  \item As this semantics does not model the underlying hardware IBT mechanism,
    it does not need to track the $\mathit{ct}$ flag, nor does it have rules for
    the $\cctarget$ instruction.
\end{itemize}

\subsection{Key Theorems}
\label{sec:key-theorems}

Our proof of relative security hinges on two key lemmas.
The first of these is a backwards compiler correctness (BCC) result, which shows
that a program hardened by \Triosecuris executed in the unrestricted speculative
semantics behaves like the original source program in the idealized semantics,
in other words, the \Triosecuris transformation correctly implements the restrictions
specified by the ideal semantics.

\begin{lemma}[Backwards compiler correctness of \Triosecuris]
\label{lem:specibt-bcc}

\begingroup\footnotesize
\setcounter{equation}{0}
\renewcommand{\theHequation}{lem\thelemma.\theequation}

\begin{alignat}{2}
& \left(\begin{array}{@{}l@{}}
  \wfp{\prog} \wedge \callee, \msf \notin \UsedVars{\prog} \wedge {}
  \\
  \isct{\prog[0]} \wedge \cctarget \notin \prog \wedge {}
  \\
  \rho'(\callee) = \vfp{(0, 0)} \wedge \rho'(\msf) = 0
  \end{array}\right)
 &\;&\Rightarrow
\label{eq:specibt-bcc-sidecond}
\\
&\tr{\prog} \vdash \runningstate{\scfg{(0,0), \rho', \mu', \emptylist}{\mtrue}{\mfalse}}\specmulti{\Obss'}{\Dirs'}\cdot
   &&\Rightarrow
\label{eq:specibt-bcc-spec}
\\
& \rho \sim \rho' \land \mu \sim \mu' &&\Rightarrow
\label{eq:specibt-bcc-match}
\\
& \exists \Obss, \Dirs.\; \prog \vdash \runningstate{\icfg{(0,0), \rho, \mu, \emptylist}{\mfalse}}
      \idealmulti{\Obss}{\Dirs}\cdot \land \Obss \sim \Obss' \land \Dirs \sim \Dirs'
&&\label{eq:specibt-bcc-ideal}
\end{alignat}

\endgroup
\end{lemma}
  This lemma connects \enquote{target runs} of the program hardened with \Triosecuris
  \eqref{eq:specibt-bcc-spec} to corresponding \enquote{source runs} of the original program
  in the ideal semantics.
  \remove{
    The initial program states are identical except for the $\msf$ and $\callee$
    registers in the register assignment \eqref{eq:specibt-bcc-match}, which are set to $\mfalse$ and $\vfp{0}$ in the target
    execution, but do not exist in the ideal state, and the $\ct$ flag, which also
    does not exist in the ideal state.
  }
  \newreviews[REQ2]{
    We define a relation $\sim$ between values in ideal and speculative program
    states. Program pointers are related by adjusting the offset according to
    the \Triosecuris transformation,\footnote{We only handle program pointers which
      either have offset zero (function pointers) or which are valid return
    addresses, since no other program pointers can occur during execution.} all
    other values are only related to themselves.
    We lift this relation to memories \eqref{eq:specibt-bcc-match} in the straightforward way and to register
    assignments by ignoring the $\msf$ and $\callee$ registers, which do not
    exist in the ideal state.

    Similarly, we lift this relation to directives and observations
    \eqref{eq:specibt-bcc-ideal}. For directives, we only need to apply it to
    $\dret\cdot$ directives, since the correct return addresses are different
    for the source program. It is not necessary to apply this mapping to
    $\dcall\cdot$ directives, since those only fall into one of two categories:
  }
  \begin{itemize}[leftmargin=*,nosep,label=$\blacktriangleright$]
  \item If the call directive points to a function entry point in the target
    execution, then it will have an offset of zero. This matches the function
    entry point in the source program, so the simulation proceeds without issue.
  \item Otherwise, the target execution will reach the $\faultstate$ state, as
    the location in the directive does not contain a $\cctarget$
    instruction. This is independent of the actual instruction at that location,
    so a mismatch between the precise location in the source and target programs
    is inconsequential.
  \end{itemize}

  \newtext{\autoref{lem:specibt-bcc}} requires a number of side conditions \eqref{eq:specibt-bcc-sidecond}:
  For one, we require that the input program does not contain any $\cctarget$
  instructions and begins with a function as the entry point (\IE $(0,0)$). The program must
  also be well-formed in the following sense: all basic blocks end in a $\cret$ or $\ijump$ instruction;
  programs and blocks are non-empty; all $\ijump$ and $\ibranch$ instructions go to labels that are not marked as function entry points;
  and all code pointer constants point to function entry points.
  Further, the source program must not use the reserved $\msf$ and $\callee$ registers, so as to not interfere with the \Triosecuris transformation.
  In the target initial state, the $\msf$ register is initialized to $0$,
  to align with the misspeculation flag in the initial state, and the $\callee$ register is initialized to $\vfp{(0, 0)}$,
  reflecting that basic block $0$ serves as the program entry point (analogous to \texttt{main}).

  \remove{
    This statement relates ideal and target executions using \emph{the same sequence
    of attacker directives}\eqref{eq:specibt-bcc-ideal}, despite the fact that the \Triosecuris transformation
    prepends two instructions to each procedure entry block, resulting in a
    discrepancy between the program counters of the source and target.
    This works because all directives fall into one of two categories, making the
    actual value of the offset largely irrelevant:
  }

\begin{proof}[Proof outline]
  \remove{The proof proceeds by strong induction on the number of steps of the target run, with
  a slightly generalized statement: Instead of the fixed initial state in
  \eqref{eq:specibt-bcc-spec} and \eqref{eq:specibt-bcc-ideal},
  they are allowed to vary, with a relation between source and
  target states ensuring that the program counter values (including on the
return stack) are properly adjusted for \Triosecuris.}
\newtext{
  The proof uses a backwards simulation with source stuttering~\cite{CompCert, Chappe26}
  with the following simulation relation:
  Generally, we relate states based on the program counters: If the program
  counter in the ideal state points to an instruction $i$, the program counter
  in the speculative state must point to beginning of the sequence
  $\tr{i}$.
  Additionally, the
  relation $\sim$ must hold for the register assignments, memories and the return
  stack, and we require that the misspeculation flag register $\msf$ in the
  speculative register assignment agrees with the misspeculation flag $\ms$ in
  the ideal state.

  For source instructions which correspond to more than one instruction in the
  hardened program, we allow the source to stutter by introducing additional
  cases, which relate all intermediate speculative states to the same ideal state:
  \begin{itemize}[leftmargin=*,nosep,label=$\blacktriangleright$]
    \item For branch instructions, the number of intermediate instructions depends on
      the branch outcome. If the branch is taken
      (1) the branch instruction to the edge-split block is already related to
      the source instruction based on the program counter, and we introduce
additional cases for
      (2) the update of the misspeculation register in the edge-split block, and
      (3) the jump instruction to the original target.
      Otherwise, the branch falls through, and we only require one additional
      case for the update of the misspeculation flag.
    \item For call instructions, a source instruction corresponds to four target instructions:
      (1) update of the $\callee$ register, (2) call to the target,
      (3) $\cctarget$ instruction at the target, and (4) update of the misspeculation register at the target.
      The first case is related to the source call instruction based on the
      program counter. We introduce extra cases relating (2) to the same call
      instruction and (3) and (4) to the following instruction, so that the leakage-producing
      $\icall$ step has a corresponding step in the source
      semantics.\jb{@YH: please check this phrasing}\yh{we defined a matching relation on states (specifically on PCs(instructions)),
      not on steps. I changed most occurrences of 'step' to 'instruction' accordingly.
      Line 2817--2818 still use 'step' because this refers the simulation consequence:
      the instruction matching makes that call steps of source and target progress in lockstep.}
    \item \newreviews[REQ2]{For return instructions, a source instruction corresponds to three target instructions:
      (1) update of the $\callee$ register with the head of the return stack,
      (2) jump to the return target, and
      (3) update of the misspeculation register at the return target.
      The first case is covered based on the program counter, and we introduce
      additional cases relating (2) to the same return instruction and (3) to the return target instruction.
    }
  \end{itemize}
  These extra cases allow us to temporarily weaken the requirement for the
  $\msf$ and return stacks to match until after the sequence has executed.
  We then prove backwards simulation between the source and target semantics with respect to this state matching relation.

  It remains to show that from the initial states in
  \eqref{eq:specibt-bcc-spec}, we can establish this simulation relation.
  However, the hardened program begins with a $\cctarget$ instruction and an
  update of the misspeculation flag, which do not correspond to any call
  instruction in the source program.
  Therefore, we handle the first two steps of the speculative execution
  separately (with the ideal execution stuttering), after which we can establish
  the simulation relation from the conditions in \eqref{eq:specibt-bcc-match}.

}

\end{proof}

An earlier version of this BCC proof was described in more detail by Leatherman-Brooks~\cite{LeathermanBrooks26}.

The other key lemma is that the ideal semantics, with its built-in protections,
enforces relative security.
We first introduce \emph{observational equivalence} to describe that an attacker
cannot distinguish two executions:

\begin{definition}[Sequential Observational Equivalence]
\label{def:def-seqobs}
Two sequential states $\runningstate{s_1}$
and $\runningstate{s_2}$ are observationally equivalent
under the sequential semantics of program $\prog$, written
$\prog \vdash \runningstate{s_1}\seqobseq\runningstate{s_2}$ iff
\[
  \forall \Obss_1,\Obss_2.\;
    \prog \vdash \runningstate{s_1} \seqmulti{\Obss_1} \cdot \land
    \prog \vdash \runningstate{s_2} \seqmulti{\Obss_2} \cdot \Rightarrow
    \Obss_1 \lessgtr \Obss_2
\]
$\Obss_1 \lessgtr \Obss_2$ denotes that, instead of the both lists being
equal,  $\Obss_1$ may be a prefix of $\Obss_2$ or vice versa. This is required
since the two executions are not required to be of the same length.
\end{definition}

\begin{definition}[Speculative Observational Equivalence]
\label{def:def-specobs}
Two speculative states $\runningstate{s_1}$
and $\runningstate{s_2}$ are observationally equivalent
under the speculative semantics of program $\prog$, written
$\prog \vdash \runningstate{s_1}\specobseq\runningstate{s_2}$ iff
\[
  \forall \Dirs,\Obss_1,\Obss_2.\;
    \prog \vdash \runningstate{s_1} \specmulti{\Obss_1}{\Dirs} \cdot \land
    \prog \vdash \runningstate{s_2} \specmulti{\Obss_2}{\Dirs} \cdot \Rightarrow
    \Obss_1 \lessgtr \Obss_2
\]
We define observational equivalence in the ideal semantics, $\idealobseq$, in
the same way.
\end{definition}

With these definitions, relative security means that sequential observational
equivalence implies observational equivalence in the ideal semantics:

\begin{lemma}[Relative Security between Sequential and Ideal Semantics]
\label{lem:specibt-ideal-rs}

\begingroup\footnotesize
\setcounter{equation}{0}
\renewcommand{\theHequation}{lem\thelemma.\theequation}

\begin{alignat}{2}
&\prog \vdash \runningstate{\pc_1, \rho_1, \mu_1, \stk_1} \, \seqobseq \, \runningstate{\pc_2, \rho_2, \mu_2, \stk_2} &\;& \Rightarrow\;
\label{eq:specibt-ideal-seqeq}
\\
&\prog \vdash \runningstate{\pc_1, \rho_1, \mu_1, \stk_1, \mfalse} \, \idealobseq \,\runningstate{\pc_2, \rho_2, \mu_2, \stk_2, \mfalse}
\label{eq:specibt-ideal-goal}
\end{alignat}
\endgroup  
\end{lemma}

\begin{proof}[Proof sketch]
  The proof relies on a property of the ideal semantics which we call
  \emph{unwinding of speculative execution}: While the misspeculation flag is
  set, any produced observations are fully determined by the program, program
  counter and return stack, \IE{} they are completely independent of registers and
  memory. Concretely, we can prove $\forall \rho_1, \mu_1, \rho_2, \mu_2.\;
  \state{\pc, \rho_1, \mu_1, \stk, \mtrue} \idealobseq \state{\pc, \rho_2,
  \mu_2, \stk, \mtrue}$ (by induction on the executions;
  equality of control flow and produced observations is ensured by masking).

  For the proof of relative security, we thus decompose both executions into two
  parts, before and after misspeculation. As the ideal behavior before
  misspeculation corresponds exactly to the sequential behavior of the program,
  we conclude by \eqref{eq:specibt-ideal-seqeq} that the first parts produce
  equal observations; further, we infer that they have equal control flow and
  trigger misspeculation at the same point, allowing us to apply the
  above-mentioned \emph{unwinding} property to conclude equality of the
  observations during the misspeculated part.
  Composing the two, we thus obtain observational equality for the ideal
  executions.
\end{proof}

By composing these two core lemmas, we obtain that \Triosecuris enforces relative
security.

\begin{theorem}[Relative Security for \Triosecuris]
\label{lem:specibt-rs}
\begingroup\footnotesize
\setcounter{equation}{0}
\renewcommand{\theHequation}{thm\thetheorem.\theequation}

\begin{alignat}{2}
& \left(\begin{array}{@{}l@{}}
  \wfp{\prog} \wedge \callee, \msf \notin \UsedVars{\prog} \wedge {}
  \\
  \isct{\prog[0]} \wedge \cctarget \notin \prog \wedge {}
  \\
  \rho_1' (\callee) = \vfp{(0, 0)} \wedge \rho_2' (\callee) = \vfp{(0, 0)} \wedge {}
  \\
  \rho_1' (\msf) = 0 \wedge \rho_2' (\msf) = 0
  \end{array}\right)
 &\;&\Rightarrow
\label{eq:specibt-rs-sidecond}
\\
&\prog \vdash \runningstate{(0,0), \rho_1, \mu_1, \emptylist} \, \seqobseq \, \runningstate{(0,0), \rho_2, \mu_2, \emptylist} &\;& \Rightarrow\;
\label{eq:specibt-rs-seqeq}
\\
& \rho_1 \sim \rho_1' \wedge \rho_2 \sim \rho_2' \land \mu_1 \sim \mu_1' \wedge \mu_2 \sim \mu_2' &&\Rightarrow
\label{eq:specibt-rs-match}
\\
&\tr{\prog} \vdash \runningstate{(0,0), \rho_1', \mu_1', \emptylist, \mtrue, \mfalse} \, \specobseq \,\runningstate{(0,0), \rho_2', \mu_2', \emptylist, \mtrue, \mfalse}
\label{eq:specibt-rs-goal}
\end{alignat}

\endgroup

\end{theorem}

\begin{proof}[Proof sketch]

The proof of relative security relies on several assumptions (\eqref{eq:specibt-rs-sidecond} and \eqref{eq:specibt-rs-match}),
which we omit here as they are identical to those described for the backwards compiler correctness lemma \autoref{lem:specibt-bcc}.

The proof itself is straightforward. We first apply \autoref{lem:specibt-bcc} to
each speculative execution in \eqref{eq:specibt-rs-goal}, yielding a pair of
corresponding ideal executions. \newtext{Since $\sim$ is injective, both ideal
executions also use the same set of directives, so by \autoref{lem:specibt-ideal-rs}, we obtain
that those ideal executions are observationally equivalent.
Speculative observational equivalence then follows from the fact that the
relation $\sim$ relates each observation in the ideal semantics to one unique speculative observation.
}
\end{proof}

\section{Translation to Machine Code}
\label{sec:translation-mc}

So far, we have discussed our \Triosecuris transformation for a language with a
basic blocks and code pointers. While this was necessary to
formalize the mitigation, it also carries the risk of inadequately modeling the
behaviors of lower-level languages, as discussed in
\autoref{sec:ki-undefined-values}. Therefore, we now demonstrate that our
security results indeed carry over to lower levels on an example compilation
step to a linearized language which we call \MiniMC.

\subsection{Definitions and Semantics}

\yh{Consider what to highlight and how to highlight it (color, font, etc).}
\begin{figure}
  \centering
  \resizebox{\columnwidth}{!}{%
  \(\begin{array}{c}
  \inferrule*[left=\textsc{McSpec\_Branch}]
  { \substack{
    \highlight{\mcfetch{\m}{\pc}{\cbranch{e}{l}}} \quad \eval{e}{\rho} = n \quad b = (n \neq 0) \\
     \pc' = \ite{b'}{\highlight{l}}{\pc+1} \quad \ms' = \ms \lor (b \neq b')
   } }
  { \runningstate{\scfg{\cfg{\pc}{\rho}{\mu}{\stk}}{\mfalse}{\ms}}
    \specstep{\obranch{b}}{\dbranch{b'}}
    \runningstate{\scfg{\cfg{\pc'}{\rho}{\mu}{\stk}}{\mfalse}{\ms'}}
  }
  \\[1.2ex]
  \inferrule*[left=\textsc{McSpec\_Jump}]
  { \highlight{\mcfetch{\m}{\pc}{\cjump{l}} } }
  { \runningstate{\scfg{\cfg{\pc}{\rho}{\mu}{\stk}}{\mfalse}{\ms}}
    \specstep{\bullet}{\bullet}
    \runningstate{\scfg{\cfg{\highlight{l}}{\rho}{\mu}{\stk}}{\mfalse}{\ms}} }
  \\[1.2ex]
  \inferrule*[left=\textsc{McSpec\_Call}]
  {
    \highlight{\mcfetch{\m}{\pc}{\ccall{e}}} \quad \eval{e}{\rho} = l \quad \ms' = \ms \lor \pc' \neq \highlight{l}
  }
  {
    \runningstate{\scfg{\cfg{\pc}{\rho}{\mu}{\stk}}{\mfalse}{\ms}}
    \specstep{\ocall{l}}{\dcall{\pc'}}
    \runningstate{\scfg{\cfg{\pc'}{\rho}{\mu}{(\pc+1)::\stk}}{\mtrue}{\ms'}}
  }
  \end{array}\)%
  }
  \caption{Speculative semantics for \MiniMC (selected rules)}
  \label{fig:spec-semantics-minimc}
  \vspace{-1em}
\end{figure}

The core aspect of \MiniMC{} is that it flattens the basic block
program structure of \MiniMIR{} into a linear list of instructions, while
leaving the syntax of instructions and expressions largely unchanged.
The \emph{linearization} step $\mc{\cdot}$\jb{Reviewer comment: where is this
defined?}\ifappendix\space(\autoref{fig:linearization-transformation})\fi, which translates \MiniMIR to
\MiniMC, simply concatenates into one list of instructions, while translating
the labels in jumps, branches and code pointer constants into the
appropriate addresses.
We assume that a \MiniMC{} program is placed in memory after the data section,
\IE the final memory layout is of the form $\m = \mu \concat \prog$, where $\mu$ and $\prog$ denote the data and code sections, respectively.

The (speculative) semantics of \MiniMC{} (\autoref{fig:spec-semantics-minimc})
is similar to the speculative semantics of \MiniMIR{}
(\autoref{fig:spec-semantics-specibt}), except for the \highlight{highlighted} changes due to the new memory layout and the absence of logical pointers\ifappendix\space(The full set of rules can be found in \autoref{sec:appendix-speculative-semantics-minimc})\fi:
\begin{itemize}[leftmargin=*,nosep,label=$\blacktriangleright$]
  \item The program counter stores an integer address, 
    rather than an offset within a basic block.
  \item Instruction fetching retrieves the instruction at the address
    specified by the program counter from the memory \m{}.
  \item Correspondingly, the targets of control-flow instructions (branches,
    jumps) are now numeric addresses in memory.
\end{itemize}

\subsection{Relative Security for Linearization}
To prove relative security proof for linearization, we employ Theorem 2
(Soundness of backwards simulations) of
\citet{OlmosBBGL25}\footnote{Our mechanization is independent of their work,
but uses the same idea.},
which states that relative security follows directly from a backwards simulation
(\IE backwards compiler correctness)
if there are functions mapping \emph{target} directives to \emph{source}
directives and \emph{source} observations to \emph{target} observations.
Intuitively, this means that the observations in the target are uniquely
determined by the observations at the source level, so if the observations of
two source executions are equal, they must result in equal observations of the
target executions.

Since linearization refines the code pointers of \MiniMIR{} into integer
addresses, we take this into account for our backwards compiler correctness
proof by defining a relation $\simeq_{\lengthof{\mu}}$ which relates code pointers
to their corresponding addresses (based on the length of the data section,
$\lengthof{\mu}$) 
, numeric values to themselves, and $\UV$ to all
numbers. We lift this relation to register assignments, memories and stacks.



\begin{lemma}[Backwards compiler correctness for Linearization]\leavevmode\vspace{-1\baselineskip}%
\label{lem:minimc-bcc}
\begingroup\footnotesize
\setcounter{equation}{0}
\renewcommand{\theHequation}{lem\thelemma.\theequation}
\begin{alignat}{2}
& \wwfp{\prog} \;\land\; \issafe{s}{spec}{\prog} \;\land\; \wfdsmc{\Dirs}{\lengthof{\mu'}}{\mc{\prog}}
 &\;&\Rightarrow
\label{eq:minimc-bcc-sidecond}
\\
&\mc{\prog} \vdash \runningstate{\scfg{\pc,\rho,\mu,\stk}{\ct}{\ms}}\specmulti{\Obss}{\Dirs}\cdot
   &&\Rightarrow
\label{eq:minimc-bcc-spec}
\\
& \pc \simeq_{\lengthof{\mu'}} \pc' \land \rho \simeq_{\lengthof{\mu'}} \rho' \land \mu \simeq_{\lengthof{\mu'}} \mu' \land \stk \simeq_{\lengthof{\mu'}} \stk' &&\Rightarrow
\label{eq:minimc-bcc-match}
\\
&\left(\begin{array}{@{}l@{}}
   \exists \Obss', \Dirs'.\; \prog \vdash \runningstate{\scfg{\pc,\rho,\mu,\stk}{\ct}{\ms}}
   \specmulti{\Obss'}{\Dirs'}\cdot \wedge
   \\
   \phantom{\exists \Obss', \Dirs'.\;}
   \Obss \simeq_{\lengthof{\mu'}} \Obss' \wedge \Dirs \simeq_{\lengthof{\mu'}} \Dirs'
   \end{array}\right)
&&\label{eq:minimc-bcc-src}
\end{alignat}
where
  {\footnotesize
  \(\quad
    \issafe{s}{spec}{\prog} \triangleq
    \begin{aligned}[t]
    &\forall \Dirs,\Obss, s'.\; \wfds{\Dirs}{\prog} ~\land~ 
    s
    \specmulti{\Obss}{\Dirs}
    \state{s'} \Rightarrow \\
    &\exists \Dirs',\Obss',s''.\;
    \state{s'}
    \specstep{\Obss'}{\Dirs'}
    s''
    \end{aligned}
  \)
  }

\endgroup
\end{lemma}

This is similar to \autoref{lem:specibt-bcc}, with the following differences:

\begin{itemize}[leftmargin=*,nosep,label=$\blacktriangleright$]
  \item Linearization refines undefined behavior (\EG if the source program
    attempts to call a numerical value, which is now interpreted as an address).
    In this case, there is no source relation to relate to, so we assume that
    the initial source state is safe under speculative execution.

    Note that we define safety only for the specific input in question, instead
    of requiring that the program is safe for all inputs. However, while
    sequential safety would typically only quantify over prefixes of one
    (deterministic) execution, safety under speculative execution actually
    quantifies over all possible execution paths (specified by different sets of
    directives $\mathcal{D}$). It would be insufficient to require only that
    there exists a safe path (this would be satisfied by the path that does not
    misspeculate).

  \item Values are related by $\simeq_{\lengthof{\mu'}}$\remove{ instead of
    equality}.
  \item We require (predicate $\wfdsmc{\Dirs}{\lengthof{\mu'}}{\mc{\prog}}$) that call directives must always point to a location within
    the code section. We justify this by the fact that the data section should
    be placed entirely in non-executable memory pages, as it would otherwise be
    impossible to provide guarantees without introducing restrictions on data.
    A similar restriction in the definition of
    $\issafe{s}{spec}{\prog}$ requires that directives use valid labels
    and offsets.
  \item \newtext{$\wwfp{\prog}$ is the weaker version of $\wfp{\prog}$ that only guarantees the following: all basic blocks end in a $\cret$ or $\ijump$ instruction; programs and blocks are non-empty.
    We proved $\forall \prog, \wfp{\prog} \rightarrow \wwfp{\tr{\prog}}$ to use this weaker condition for linearization, which is satisfied by the output of the \Triosecuris transformation.}
\end{itemize}

\begin{proof}[Proof sketch]
  The proof relies on a lock-step backwards simulation: 
  if the source state is safe and the source and target states are related,
  then for every step in the target, the source can take a corresponding step
  while preserving the relation between the states.
  We apply this simulation inductively to obtain the result for multi-step
  executions.
\end{proof}


\begin{lemma}[Relative Security for Linearization]
\label{lem:mc-rs}
\begingroup\footnotesize
\setcounter{equation}{0}
\renewcommand{\theHequation}{lem\thelemma.\theequation}

\begin{alignat}{2}
& \wwfp{\prog} \land \len = \lengthof{\mu_1'} = \lengthof{\mu_2'} &&\Rightarrow
\label{eq:mc-rs-sidecond}
\\
& \issafe{s_1}{spec}{\prog} \wedge \issafe{s_2}{spec}{\prog} &&\Rightarrow
\label{eq:mc-rs-safe}
\\
&
  \prog \vdash \runningstate{(0,0), \rho_1, \mu_1, \emptylist, \mtrue, \mfalse} \, \specobseq \,\runningstate{(0,0), \rho_2, \mu_2, \emptylist, \mtrue, \mfalse}
  &&\Rightarrow
\label{eq:mc-rs-seqeq}
\\
& \rho_1 \simeq_{\mu_1'} \rho_1' \wedge \rho_2 \simeq_{\mu_2'} \rho_2' \wedge \mu_1 \simeq_{\mu_1'} \mu_1' \wedge \mu_2 \simeq_{\mu_2'} \mu_2' &&\Rightarrow
\label{eq:mc-rs-match}
\\
& \mc{\prog}, \len \vdash \runningstate{0, \rho_1', \mu_1', \emptylist, \mtrue, \mfalse} \, \mcobseq \,\runningstate{0, \rho_2', \mu_2', \emptylist, \mtrue, \mfalse}
\end{alignat}
where
  {\footnotesize
  \[
   \prog, \len \vdash \runningstate{s} \mcobseq \runningstate{s'} \triangleq
   \begin{array}[t]{@{}l@{}}
   \forall \Dirs,\Obss_1,\Obss_2.\;
   \wfdsmc{\Dirs}{\len}{\mc{\prog}} \;\land\;
   \\
   \prog \vdash \runningstate{s}
    \specmulti{\Obss_1}{\Dirs}\cdot \;\land\;
   \prog \vdash \runningstate{s'}
   \specmulti{\Obss_2}{\Dirs}\cdot \;\Rightarrow\;
  \Obss_1 \lessgtr \Obss_2
   \end{array}
  \]
  }
  and $\mcobseq$ is defined similarly to $\specobseq$.
\endgroup

\end{lemma}

\begin{proof}[Proof sketch]
  To obtain relative security for linearization, we apply Theorem 2 of
  \citet{OlmosBBGL25} with the backwards simulation of \autoref{lem:minimc-bcc}.
  This requires that there are functions mapping target directives to source
  directives, and source directives from target directives.

  We obtain the mapping from target directives to source directives directly
  from the relation $\simeq_{\lengthof{\mu}}$, since the restriction of $\simeq_{\lengthof{\mu}}$
  to code pointer values is a bijection between code pointers and
  in-bounds addresses. 

  For observations, we similarly use the fact that $\simeq_{\lengthof{\mu}}$ is
  functional for all values other than $\UV$, which cannot occur in
  observations.
\end{proof}

\subsection{End-to-End Relative Security}

As mentioned in \autoref{sec:ki-undefined-values}, \autoref{lem:mc-rs} requires
that the (speculative) \MiniMIR executions are safe. Thus, in order to obtain
end-to-end relative security, we need to prove that this is the case.

\newtheorem*{lem:specibt-safety-preservation}{\autoref{lem:specibt-safety-preservation}}

\begin{lem:specibt-safety-preservation}[Safety Preservation for \Triosecuris]
\label{lem:specibt-safety}
\begingroup\footnotesize
\setcounter{equation}{0}
\renewcommand{\theHequation}{lem\thelemma.\theequation}

\begin{alignat}{2}
& \left(\begin{array}{@{}l@{}}
  \wfp{\prog} \wedge \callee, \msf \notin \UsedVars{\prog} \wedge {}
  \\
  \isct{\prog[0]} \wedge \cctarget \notin \prog \wedge {}
  \\
  \lengthof{\mu} > 0 \wedge \rho'(\callee) = \&0 \wedge \rho'(\msf) = 0
  \end{array}\right)
 &\;&\Rightarrow
\label{eq:specibt-safe-sidecond}
\\
& \issafe{\runningstate{(0,0), \rho, \mu, \emptylist}}{seq}{\prog}
   &&\Rightarrow
\label{eq:specibt-safe-seq}
\\
& \rho \sim \rho' \land \mu \sim \mu' &&\Rightarrow
\label{eq:specibt-safe-match}
\\
&\issafe{\runningstate{\scfg{(0,0), \rho', \mu', \emptylist}{\mtrue}{\mfalse}}}{spec}{\tr{\prog}}
   &&
\label{eq:specibt-safe-spec}
\end{alignat}
where
  {\footnotesize
  \(\quad
    \issafe{s}{seq}{\prog} \triangleq
    \forall \Obss, s'.\;
    s
    \seqmulti{\Obss}
    \runningstate{\scfg{s'}} \Rightarrow
    \exists \Obss',s''.\;
    \runningstate{\scfg{s'}}
    \seqstep{\Obss'}
    s''
  \)
  }
\endgroup
\end{lem:specibt-safety-preservation}


\begin{proof}[Proof sketch]
  We divide the proof into two parts:
  first, we prove safety preservation between sequential and ideal semantics;
  second, we prove safety preservation between ideal and speculative semantics.

  For the first part, before misspeculation, ideal and sequential semantics behave identically.
  Once the misspeculation flag is set, it cannot be reverted, and all instructions that would cause undefined behavior are masked with appropriate constant values.
  Therefore, ideal execution is always safe.

  For the second part, we use a strengthened version of \autoref{lem:specibt-bcc}
  to obtain the ideal state matching the last speculative state. Since
  the ideal execution is safe, we obtain a step of the ideal semantics. By case
  distinction on this step, we prove that the speculative semantics takes
  a corresponding step. 
\end{proof}

We have now proved all the lemmas and theorems needed to establish end-to-end
relative security, as illustrated in \autoref{fig:composition}.

\begin{theorem}[End-to-End Relative Security]
\label{lem:ete-rs}
\begingroup\footnotesize
\setcounter{equation}{0}
\renewcommand{\theHequation}{thm\thetheorem.\theequation}

\begin{alignat}{2}
& \left(\begin{array}{@{}l@{}}
  \wfp{\prog} \wedge \callee, \msf \notin \UsedVars{\prog} \wedge {}
  \\
  \isct{\prog[0]} \wedge \cctarget \notin \prog \wedge {}
  \\
  \len = \lengthof{\mu_1'} = \lengthof{\mu_2'} > 0  \wedge {}
  \\
  \rho_1'(\callee) = \len \wedge \rho_2'(\callee) = \len \wedge {}
  \\
  \rho_1'(\msf) = 0 \wedge \rho_2'(\msf) = 0
  \\
  \wfreg{\rho_1} \land \wfreg{\rho_2} \land \wfmem{\mu_1} \land \wfmem{\mu_2}
  \end{array}\right)
 &\;&\Rightarrow
\label{eq:ete-rs-sidecond}
\\
& \issafe{\runningstate{(0,0), \rho_1, \mu_1, \emptylist}}{seq}{\prog} \wedge \issafe{\runningstate{(0,0), \rho_2, \mu_2, \emptylist}}{seq}{\prog} &&\Rightarrow
\label{eq:ete-rs-safe}
\\
&\prog \vdash \runningstate{\pc_1, \rho_1, \mu_1, \stk_1} \, \seqobseq \, \runningstate{\pc_2, \rho_2, \mu_2, \stk_2} &\;& \Rightarrow\;
\label{eq:ete-rs-seqeq}
\\
& \rho_1 \cong_{\lengthof{\mu_1'}} \rho_1' \wedge \rho_2 \cong_{\lengthof{\mu_2'}} \rho_2' \wedge \mu_1 \simeq_{\lengthof{\mu_1'}} \mu_1' \wedge \mu_2 \simeq_{\lengthof{\mu_2'}} \mu_2' &&\Rightarrow
\label{eq:ete-rs-match}
\\
& \mc{\tr{\prog}}, \len \vdash \runningstate{0, \rho_1', \mu_1', \emptylist, \mtrue, \mfalse} \, \mcobseq \,\runningstate{0, \rho_2', \mu_2', \emptylist, \mtrue, \mfalse}
\label{eq:ete-rs-mceq}
\end{alignat}

\endgroup

\end{theorem}

$\cong_{\lengthof{\mu}}$ denotes the composition of $\sim$ and $\simeq_{\lengthof{\mu}}$.
\newtext{$\wfreg{\rho}$ and $\wfmem{\mu}$ are well-formedness conditions for register assignments and memories,
which require that every code pointer stored in registers and memory, respectively, is either a function entry pointer or a valid return address.}
The remaining side conditions match \autoref{lem:specibt-safety-preservation},
\autoref{lem:mc-rs} and \autoref{lem:specibt-rs}.

\begin{proof}[Proof sketch]
  We construct intermediate register assignments $\rho^i_1, \rho_1 \sim
  \rho^i_1$ and $\rho^i_2, \rho_2 \sim \rho^i_2$ by updating both assignments
  with $\msf = 0$ and $\callee = \vfp{(0, 0)}$.
  By \autoref{lem:specibt-rs} with \eqref{eq:ete-rs-seqeq}, we obtain
  observational equivalence of  $\runningstate{\scfg{(0,0), \rho^i_1, \mu_1,
  \emptylist}{\mtrue}{\mfalse}}$ and $\runningstate{\scfg{(0,0), \rho^i_2,
  \mu_2, \emptylist}{\mtrue}{\mfalse}}$.
  By \autoref{lem:specibt-safety-preservation} and \eqref{eq:ete-rs-safe}, we
  further obtain that all executions starting in those states are safe.

  From $\rho_1 \sim \rho^i_1$ (resp.\ $\rho_2 \sim \rho^i_2$) and the conditions
  on $\rho_1', \rho_2'$ in \eqref{eq:ete-rs-sidecond}, we obtain $\rho^i_1
  \simeq_{\lengthof{\mu_1'}} \rho_1'$ (resp.\ $\rho^i_2 \simeq_{\lengthof{\mu_2'}} \rho_2'$).
  We conclude by \autoref{lem:mc-rs}.
\end{proof}



\section{Related Work}
\label{sec:related-work}

\paragraph{Spectre BTB \newreviews[REQ2]{and RSB} defenses}
We are not the first to propose defenses against Spectre BTB \newreviews[REQ2]{and RSB}, with previous work
in this space including Serberus~\cite{MosierNMT24},
Swivel~\cite{NarayanDMCJGVSS21}, ShadowCFI~\cite{TrujilloKKY} and the work of \citet{FabianPGB25}.
%
Serberus~\cite{MosierNMT24} offers the most complete Spectre mitigation to
date, protecting not only against Spectre BTB, PHT, and RSB, but even against
the \emph{data speculation} variants Spectre STL and PSF.
\newreviews[REQ3]{It uses only coarse-grained IBT, accounting for the resulting
  imprecision by constructing a transient
control flow graph according to the remaining attacker capabilities. Serberus
then relies on}
\removereviews{It achieves this with }a combination of fence
insertion along a minimum cut, function-private stacks, and register cleaning
to obtain a comprehensive yet efficient mitigation.
Their work includes high-level security proofs on paper,
while our work is fully mechanized in Rocq.
Their setting is more restrictive than the standard cryptographic
constant-time discipline, additionally requiring that all function arguments and
return values must be public (passing secrets is thus only possible by
reference). In contrast, our mitigation is much more permissive,
protecting arbitrary programs.
\removereviews{Finally, as mentioned in \autoref{sec:intro}, by precisely detecting
misspeculation we also obtain simplicity benefits for \Triosecuris.}


\newreviews[REQ3]{Finally, we wonder if fine-grained or precise IBT could also benefit Serberus.
  A reduced transient control flow graph might in principle reduce the
  number of fences Serberus needs, and perhaps even relax its register-clearing
  pass or the requirement that all arguments be public---but whether this
  can be achieved in the relatively complex Serberus design is unclear to us.
  The simple way \Triosecuris uses CET to precisely detect misspeculation
  and set the SLH flag does not directly apply to Serberus, which is
  not based on SLH, but on adding fences.}

Swivel~\cite{NarayanDMCJGVSS21} protects Wasm modules against Spectre PHT, BTB,
and RSB mounted from other modules running in the same process.
Their attack model is different from ours, as it aims to prevent malicious Wasm
modules both (1)~from speculatively breaking out of the Wasm sandbox and
(2)~from inducing speculative leakage in other modules.
Part (2) is similar to our attacker model, but weaker, since
they consider an attack mitigated if the attacker cannot train the prediction mechanism to take a
specific path, whereas our work assumes that the attacker has perfect control
over the prediction mechanisms, which prevents more sophisticated
attacks~\cite{WiebingG25}.
\newreviews[REQ3]{Their leakage model is also weaker, including only memory
  accesses to mapped memory regions.

One of their enforcement techniques, included in the Swivel-CET variant, also
uses software checks on top of hardware IBT. Unlike our technique, which makes
use of the code pointers directly, they rely on block identifiers and do not
cover return edges.
Their approach does not use SLH and instead relies on
clearing the heap base pointer to restrict all memory accesses to unmapped
memory regions. While sufficient under the restrictions of Wasm and their leakage model, this approach
would not provide adequate protection for more general programs.
}
%
Finally, while a small formal model of Swivel was sketched \newreviews[REQ3]{in a short paper} by
\citet{CauligiGMSV22}, \newreviews[REQ3]{their modeling of Spectre RSB attacks is unrealistic}
and security proofs were left for future work.

The comparison of code pointers that \Triosecuris performs at procedure entries
has a precursor in the Spectre mitigations that \citet{McIlroySTTV19} describe
for a version of the V8 JavaScript engine.
There a poison register, updated at conditional branches like SLH's
misspeculation flag, masks the addresses of loads, for a 19\% slowdown on the
Octane JavaScript benchmark.
For indirect branches, a reserved register holds the intended target, and the
code at every possible target compares it, without branching, with its own
address and updates the poison register.
For this instrumentation they report negligible overhead on optimized code, and
a 12\% slowdown with V8's interpreter alone.
They note that an attacker can bypass the comparison by steering the branch to a
location that is not an instrumented target.
In our setting CET rules out such locations, which makes our detection precise.
Building on Ultimate SLH, \Triosecuris masks not only the addresses of loads
but every operand that could leak, as protecting arbitrary programs requires.
\Triosecuris also applies the same kind of comparison to returns and has a
mechanized security proof.

In very recent independent work, \citet{TrujilloKKY} proposed Register Hiding and ShadowCFI, a mitigation against Spectre BTB and RSB
 which also precisely detects misspeculation, without requiring CET.
However, so-called ``architectural register-state
independent secret reachability attacks'' are out of scope for them, whereas
our mitigation does protect against such attacks, while also offering protection
against Spectre PHT.
The paper provides an argument for a simplified version of their
mitigation, but this argument is informal (not really a formal security proof) and it does not model memory.

\citet{FabianPGB25} introduce a secure compilation framework that allows lifting
the security guarantees provided by Spectre countermeasures from weaker
speculative semantics to stronger ones.
This allows them to take software defenses against Spectre BTB (retpoline
variants~\cite{retpoline}), prove them secure in a semantics that only includes
BTB, and obtain security against more variants (\EG PHT), for programs that are
already secure against the extra variants (\EG PHT).
They see this as a first step towards more easily composing Spectre mitigations.
They also lift PHT defenses like Ultimate SLH to obtain security against other
variants (\EG RSB, STL, SLS), for programs that are already secure against the
extra variants (\EG RSB, STL, SLS).
They show, however, that without modifications SLH-based defenses cannot be
lifted in the same way to obtain security against BTB, as indirect jumps can be
used to bypass SLH protection.
In this paper we show that Ultimate SLH can be modified to also defend against
BTB and for this we combine it with CET-style protection, a defense that is not
investigated by \citet{FabianPGB25}.

\paragraph{Spectre RSB defense}
\citet{OlmosBCGLOSYZ25} propose a mitigation for Spectre RSB and PHT for
the Jasmin language~\cite{AlmeidaBBBGLOPS17} and do a security proof for a
simplified version of it in Rocq.  They compile returns to
tables of conditional jumps, then use \emph{selective} SLH
\cite{ShivakumarBBCCGOSSY23} to protect against Spectre PHT. Their work is
targeted only at cryptographic constant time code, and protection against
Spectre BTB is trivial, as their source language does not support indirect
jumps, and their transformation does not insert them.

\paragraph{Relative security and preservation of constant time by compilation}

In this work we use the relative security definition for compilers introduced by
\citet{BaumannBDHH25}.
%
This is inspired by \citet{DongolGPW24}, who defined relative security for a single program, but
varying inputs.
This is also inspired by related concepts used in prior work: speculative
non-interference~\cite{GuarnieriKMRS20, PatrignaniG21}, relative
non-interference~\cite{ShivakumarBBCCGOSSY23, CauligiDMBS22}, relative
constant-time~\cite{ZhangBCSY23}.
%

For compiling cryptographic code, a common security guarantee is
\emph{preservation of constant-time}~\cite{OlmosBBGL25, BartheBGHLPT20}\jb{there
should be plenty more if we want},
\IE proving that the target program is constant-time if
the source program is. This differs from relative security since it makes
certain assumptions about which inputs must be indistinguishable in the source
program, whereas relative security simply preserves indistinguishability.
However, it is worth noting that generalized versions, such as $\phi$-SCT
\cite{OlmosBBGL25}, can preserve arbitrary equivalence relations, and are thus
equivalent to relative security. One difference is that \citet{OlmosBBGL25}
require safety of
the source program, whereas our formalization allows for the program to be
unsafe on some inputs.

\section{Conclusion and Future Work}
\label{sec:conclusion}\label{sec:future-work}

In this paper we introduced \Triosecuris, a formally verified defense that protects
arbitrary programs against Spectre attacks by combining CET-style
hardware-assisted control-flow integrity with compiler-inserted Ultimate SLH.
\paragraph{Future LLVM Implementation}
In the future, we would like to implement \Triosecuris in LLVM by extending the
implementation of Ultimate SLH~\cite{ZhangBCSY23}.
\newreviews[REQ1]{
Since the features that \Triosecuris relies on---especially the CET
instructions---are already exposed at the LLVM MIR level,
where Ultimate SLH is implemented, we expect no fundamental obstacles.

We expect the additional overhead introduced by \Triosecuris over that of Ultimate SLH to be small.
The extra overhead of our instrumentation should mainly come from two sources:
(1)~register pressure from passing the callee and return address,
but we expect this to be negligible, as the live range of these registers is confined to the control-transfer
boundary---starting before the indirect transition and ending after the
check, which is performed right after the transition;
(2)~detection code inserted at call and return sites, but these happen less
frequently than the branches Ultimate SLH already instruments, and the
overhead of this instrumentation is anyway low relative to the masking
overhead of Ultimate SLH~\cite[Figure 3]{ZhangBCSY23}.
Given the performance evaluation of Ultimate SLH~\cite[Figure 3]{ZhangBCSY23},
we also expect \Triosecuris to incur less overhead than lfence-based defenses against Spectre PHT.
An implementation would allow empirically checking these expectations
by running experiments using standard benchmarks.

While we expect the additional overhead of \Triosecuris to be small,
Ultimate SLH itself has a significant overhead of $\sim$150\% on the SPEC benchmarks~\cite{ZhangBCSY23}.}
\newreviews[REQ4B1]{To reduce this overhead, it would be interesting to investigate whether \Triosecuris can be built on top of a
selective version of Ultimate SLH called FSLH~\cite{BaumannBDHH25}.}
\newreviews[REQ1]{This may be nontrivial though,
because the formal guarantees of \Triosecuris currently rely on the comprehensive masking of Ultimate SLH
(\EG to prevent undefined behavior for safety preservation).
Moreover, implementing FSLH itself in LLVM seems challenging~\cite{BaumannBDHH25}.

Finally, the Ultimate SLH implementation allows linking with unprotected code,
but this weakens the security guarantees and hasn't been formally studied.
Preserving the sequential correctness of \Triosecuris' instrumentation
when linking with unprotected code would also be challenging to achieve.}

\newtext{
\paragraph{More precisely modeling the stack.}
Inspired by Serberus~\cite{MosierNMT24}, our semantics relies on a separate call
stack (\autoref{sec:lang}), an abstraction of the normal stack in memory and a
CET-protected shadow stack.
While this is a convenient simplification for our proofs, this may be relying on
hardware protection unnecessarily.
Since our current version of \Triosecuris is built on Ultimate SLH, this already
masks all speculative memory writes, including all speculative stack overwrites.
Assuming a CET shadow stack thus does not seem necessary in this context, and more
precisely modeling the layout of the normal stack in memory would allow one to
prove this in the future.
Moreover, more precisely modeling the stack would also allow one
to investigate whether \Triosecuris can also make use of
software defenses such as retpoline~\cite{retpoline} and saferet~\cite{saferet},
for old processors without Intel CET or Arm BTI, or for which these hardware
features are not implemented in a way that is secure under speculation.
}

\paragraph{Other indirect jumps}

Finally, it could be interesting to extend \Triosecuris to indirect jumps that go
beyond calls through function pointers (e.g. switch tables).
\jb{I'd argue we can drop this paragraph if we need space}

\ch{Included above:  The one thing we should definitely mention here, probably before all
  other extensions, is that in the presence of Ultimate SLH we may not need the
  CET shadow stack. Work is already underway on this, right?}

\ch{Included above: It may be interesting to model the CET shadow stack more concretely.
  In our post-CSF discussion we identified this as the most interesting extension
  to try to investigate immediately. Board of that whole discussion:
  \url{https://photos.app.goo.gl/UQZKVBBjV7B1XAFi7}}

\jb{It could be interesting here that (according to the SpecCFI paper), ARM BTI
  has different target markers for jumps and calls}\ch{what's the difference?
  would it change the high-level story above?}
  \jb{no, this comment is quite old. this would be a technical detail for an
  extension to indirect jumps, not something worth mentioning here}

\ch{Unsure it fits above (risk of the reviewers asking why we didn't do it
  already, with proofs): Using testing to make our machine more realistic
  without going all the way to LLVM: RISC-V instead of MiniMC.
  There may be something interesting happening with register
  allocation~\cite{VanDerWallM25}.}

\ch{Don't think we should mention it here; but relating to something like
  MiniRevizor (always mispredict) could also be interesting.
  It was on pre-CSF our board, but afterwards we said it's not so interesting
  (orthogonal to everything).}


\ifanon\else
{\small
\paragraph{Acknowledgments}
We are grateful to Santiago Arranz Olmos, Lucie Lahaye,
Sören van der Wall, Yuval Yarom, and
Zhiyuan Zhang for the insightful discussions.
We also thank the CSF'26 reviewers and shepherd for their very helpful feedback.
This work was in part funded by the
Deutsche Forschungsgemeinschaft (DFG, German Research Foundation)
under Germany's Excellence Strategy -- EXC 2092 CASA -- 390781972.
This work was also in part funded by the Institute of Information \&
Communications Technology Planning \& Evaluation (IITP) grant funded by the
Korea government (MSIT) (No.RS-2024-00441762, Global Advanced Cybersecurity
Human Resources Development).
\ch{TODO: add more acknowledgements here, if needed}
}
\fi


\ifappendix
\appendices

\section{Full Definition for \Triosecuris}
\label{sec:appendix-triosecuris-full}

\subsection{Sequential Semantics for \MiniMIR}
\label{sec:appendix-seq-semantics}

\autoref{fig:seq-semantics-specibt} shows the full sequential semantics for \MiniMIR.
\begin{figure*}
  \centering
  \[
  \inferrule*[left=\textsc{Seq\_Skip}]
    { \fetch{\prog}{\pc}{\cskip} }
    { \runningstate{\cfg{\pc}{\rho}{\mu}{\stk}}
      \seqstep{\bullet}
      \runningstate{\cfg{\pc+1}{\rho}{\mu}{\stk}}
    }
  \]\[
  \inferrule*[left=\textsc{Seq\_Asgn}]
    { \fetch{\prog}{\pc}{\casgn{x}{e}} \quad \eval{e}{\rho} = v }
    { \runningstate{\cfg{\pc}{\rho}{\mu}{\stk}}
      \seqstep{\bullet}
      \runningstate{\cfg{\pc+1}{\subst{x}{v}{\rho}}{\mu}{\stk}}
    }
  \]\[
  \inferrule*[left=\textsc{Seq\_Div}]
    { \fetch{\prog}{\pc}{\cdiv{x}{e_1}{e_2}} \quad \eval{e_1}{\rho} = n_1 \quad \eval{e_2}{\rho} = n_2 \quad
    n = {\begin{cases} n_1/n_2 & \text{if } n_2 \neq 0 \\ \UV & \text{if } n_2 = 0 \end{cases}}
    }
    { \runningstate{\cfg{\pc}{\rho}{\mu}{\stk}}
      \seqstep{\odiv{n_1}{n_2}}
      \runningstate{\cfg{\pc+1}{\subst{x}{n}{\rho}}{\mu}{\stk}}
    }
  \]\[
  \inferrule*[left=\textsc{Seq\_Branch}]
  { \substack{
    \fetch{\prog}{\pc}{\cbranch{e}{l}} \quad \eval{e}{\rho} = n \quad b = (n \neq 0) \\
     \pc' = \ite{b}{(l,0)}{\pc+1}
   } }
  { \runningstate{\cfg{\pc}{\rho}{\mu}{\stk}}
    \seqstep{\obranch{b}}
    \runningstate{\cfg{\pc'}{\rho}{\mu}{\stk}}
  }
  \]\[
  \inferrule*[left=\textsc{Seq\_Jump}]
  { \fetch{\prog}{\pc}{\cjump{l}} }
  { \runningstate{\cfg{\pc}{\rho}{\mu}{\stk}}
    \seqstep{\bullet}
    \runningstate{\cfg{(l,0)}{\rho}{\mu}{\stk}} }
  \]\[
  \inferrule*[left=\textsc{Seq\_Load}]
  { \fetch{\prog}{\pc}{\cload{x}{e}} \quad
    \eval{e}{\rho} = n \quad \mu [n] = v
  }
  { \runningstate{\cfg{\pc}{\rho}{\mu}{\stk}}
    \seqstep{\oload{n}}
    \runningstate{\cfg{\pc+1}{\subst{x}{v}{\rho}}{\mu}{\stk}}
  }
  \]\[
  \inferrule*[left=\textsc{Seq\_Store}]
  {
    \fetch{\prog}{\pc}{\cstore{e}{e'}} \quad \eval{e}{\rho} = n \quad \eval{e'}{\rho} = v
  }
  {
    \runningstate{\cfg{\pc}{\rho}{\mu}{\stk}}
    \seqstep{\ostore{n}}
    \runningstate{\cfg{\pc+1}{\rho}{\subst{n}{v}{\mu}}{\stk}}
  }
  \]\[
  \inferrule*[left=\textsc{Seq\_Call}]
  {
    \fetch{\prog}{\pc}{\ccall{e}} \quad \eval{e}{\rho} = \vfp{(l, o)}
  }
  {
    \runningstate{\cfg{\pc}{\rho}{\mu}{\stk}}
    \seqstep{\ocall{(l, o)}}
    \runningstate{\cfg{(l, o)}{\rho}{\mu}{(\pc+1)::\stk}}
  }
  \]\[
  \inferrule*[left=\textsc{Seq\_Peek}]
  { \fetch{\prog}{\pc}{\cpeek{x}} \quad
    v = {\begin{cases} \pc' & \text{if } \stk = \pc' :: \stk' \\ \UV & \text{if } \stk = \epsilon \end{cases}}
  }
  { \runningstate{\cfg{\pc}{\rho}{\mu}{\stk}}
    \seqstep{\bullet}
    \runningstate{\cfg{\pc+1}{\subst{x}{v}{\rho}}{\mu}{\stk}}
  }
  \]\[
  \inferrule*[left=\textsc{Seq\_Ret}]
  {
    \fetch{\prog}{\pc}{\cret}
  }
  {
    \runningstate{\cfg{\pc}{\rho}{\mu}{\pc' :: \stk}}
    \seqstep{\bullet}
    \runningstate{\cfg{\pc'}{\rho}{\mu}{\stk}}
  }
  \]\[
  \inferrule*[left=\textsc{Seq\_Term}]
  {
    \fetch{\prog}{\pc}{\cret}
  }
  {
    \runningstate{\cfg{\pc}{\rho}{\mu}{\emptylist}}
    \seqstep{\bullet}
    \termstate
  }
  \]
  \caption{Sequential semantics for \MiniMIR}
  \label{fig:seq-semantics-specibt}
  \vspace{-1em}
\end{figure*}
\label{app-1}

\subsection{Ideal Semantics for \MiniMIR}
\label{sec:appendix-ideal-semantics}

\autoref{fig:ideal-semantics-specibt-full} shows the full ideal semantics for \MiniMIR.

\begin{figure*}
  \centering
  \[
  \inferrule*[left=\textsc{Ideal\_Skip}]
    { \fetch{\prog}{\pc}{\cskip} }
    { \runningstate{\icfg{\cfg{\pc}{\rho}{\mu}{\stk}}{\ms}}
      \idealstep{\bullet}{\bullet}
      \runningstate{\icfg{\cfg{\pc+1}{\rho}{\mu}{\stk}}{\ms}}
    }
  \]\[
  \inferrule*[left=\textsc{Ideal\_Asgn}]
    { \fetch{\prog}{\pc}{\casgn{x}{e}} \quad \eval{e}{\rho} = v }
    { \runningstate{\icfg{\cfg{\pc}{\rho}{\mu}{\stk}}{\ms}}
      \idealstep{\bullet}{\bullet}
      \runningstate{\icfg{\cfg{\pc+1}{\subst{x}{v}{\rho}}{\mu}{\stk}}{\ms}}
    }
  \]\[
  \inferrule*[left=\textsc{Ideal\_Div}]
    { \substack{
      \fetch{\prog}{\pc}{\cdiv{x}{e_1}{e_2}} \quad \highlight{\ite{\ms}{0}{\eval{e_1}{\rho}} = n_1} \\
      \highlight{\ite{\ms}{0}{\eval{e_2}{\rho}} = n_2} \quad
      n = {\begin{cases} n_1/n_2 & \text{if } n_2 \neq 0 \\ \UV & \text{if } n_2 = 0 \end{cases}}
    }
    }
    { \runningstate{\icfg{\cfg{\pc}{\rho}{\mu}{\stk}}{\ms}}
      \idealstep{\odiv{n_1}{n_2}}{\bullet}
      \runningstate{\icfg{\cfg{\pc+1}{\subst{x}{n}{\rho}}{\mu}{\stk}}{\ms}}
  }
  \]\[
  \inferrule*[left=\textsc{Ideal\_Branch}]
  { \substack{
    \fetch{\prog}{\pc}{\cbranch{e}{l}} \quad \highlight{\ite{\ms}{0}{\eval{e}{\rho}} = n} \quad b = (n \neq 0) \\
     \pc' = \ite{b'}{(l,0)}{\pc+1} \quad \ms' = \ms \lor (b \neq b')
   } }
  { \runningstate{\icfg{\cfg{\pc}{\rho}{\mu}{\stk}}{\ms}}
    \idealstep{\obranch{b}}{\dbranch{b'}}
    \runningstate{\icfg{\cfg{\pc'}{\rho}{\mu}{\stk}}{\ms'}}
  }
  \]\[
  \inferrule*[left=\textsc{Ideal\_Jump}]
  { \fetch{\prog}{\pc}{\cjump{l}} }
  { \runningstate{\icfg{\cfg{\pc}{\rho}{\mu}{\stk}}{\ms}}
    \idealstep{\bullet}{\bullet}
    \runningstate{\icfg{\cfg{(l,0)}{\rho}{\mu}{\stk}}{\ms}} }
  \]\[
  \inferrule*[left=\textsc{Ideal\_Load}]
  { \fetch{\prog}{\pc}{\cload{x}{e}} \quad \highlight{\eval{\ite{ms}{0}{e}}{\rho} = n} \quad \mu [n] = v
  }
  { \runningstate{\icfg{\cfg{\pc}{\rho}{\mu}{\stk}}{\ms}}
    \idealstep{\oload{n}}{\bullet}
    \runningstate{\icfg{\cfg{\pc+1}{\subst{x}{v}{\rho}}{\mu}{\stk}}{\ms}}
  }
  \]\[
  \inferrule*[left=\textsc{Ideal\_Store}]
  {
    \fetch{\prog}{\pc}{\cstore{e}{e'}} \quad \highlight{\eval{\ite{ms}{0}{e}}{\rho} = n} \quad \eval{e'}{\rho} = v
  }
  {
    \runningstate{\icfg{\cfg{\pc}{\rho}{\mu}{\stk}}{\ms}}
    \idealstep{\ostore{n}}{\bullet}
    \runningstate{\icfg{\cfg{\pc+1}{\rho}{\subst{n}{v}{\mu}}{\stk}}{\ms}}
  }
  \]\[
  \inferrule*[left=\textsc{Ideal\_Call}]
  {
    \fetch{\prog}{\pc}{\ccall{e}} \quad \highlight{\ite{\ms}{\vfp{(0, 0)}}{\eval{e}{\rho}} = \vfp{(l, o)}} \quad \ms' = \ms \lor (l', o') \neq (l, o) \\
    \highlight{\lengthof{\prog} > l' \quad \isct{\prog[l']}} \quad \highlight{o' = 0}
  }
  {
    \runningstate{\icfg{\cfg{\pc}{\rho}{\mu}{\stk}}{\ms}}
    \idealstep{\ocall{(l, o)}}{\dcall{(l', o')}}
    \runningstate{\icfg{\cfg{\pc'}{\rho}{\mu}{(\pc+1)::\stk}}{\ms'}}
  }
  \]\[
  \inferrule*[left=\textsc{Ideal\_Call\_Fault}]
  {
    \fetch{\prog}{\pc}{\ccall{e}} \quad \highlight{\ite{\ms}{\vfp{(0, 0)}}{\eval{e}{\rho}} = \vfp{(l, o)}} \\
    \highlight{\lengthof{\prog} \leq l' \lor \neg\isct{\prog[l']} \lor (o' \neq 0)}
  }
  {
    \runningstate{\icfg{\cfg{\pc}{\rho}{\mu}{\stk}}{\ms}}
    \idealstep{\ocall{(l, o)}}{\dcall{(l', o')}}
    \faultstate
  }
  \]\[
  \inferrule*[left=\textsc{Ideal\_Peek}]
  { \fetch{\prog}{\pc}{\cpeek{x}} \quad
    v = {\begin{cases} \pc' & \text{if } \stk = \pc' :: \stk' \\ \UV & \text{if } \stk = \epsilon \end{cases}}
  }
  { \runningstate{\icfg{\cfg{\pc}{\rho}{\mu}{\stk}}{\ms}}
    \idealstep{\bullet}{\bullet}
    \runningstate{\icfg{\cfg{\pc+1}{\subst{x}{v}{\rho}}{\mu}{\stk}}{\ms}}
  }
  \]\[
  \inferrule*[left=\textsc{Ideal\_Ret}]
  {
    \fetch{\prog}{\pc}{\cret} \quad \wfret{\pc''} \quad \ms' = \ms \lor \pc' \neq \pc''
  }
  {
    \runningstate{\icfg{\cfg{\pc}{\rho}{\mu}{\pc' :: \stk}}{\ms}}
    \idealstep{\bullet}{\dret{\pc''}}
    \runningstate{\icfg{\cfg{\pc''}{\rho}{\mu}{\stk}}{\ms'}}
  }
  \]\[
  \inferrule*[left=\textsc{Ideal\_Term}]
  {
    \fetch{\prog}{\pc}{\cret}
  }
  {
    \runningstate{\icfg{\cfg{\pc}{\rho}{\mu}{\emptylist}}{\ms}}
    \idealstep{\bullet}{\bullet}
    \termstate
  }
  \]
  \caption{Ideal semantics for \MiniMIR}
  \label{fig:ideal-semantics-specibt-full}
  \vspace{-1em}
\end{figure*}

\section{Full Definition for Linearization}
\label{sec:appendix-linearization}

\subsection{Speculative Semantics for \MiniMC}
\label{sec:appendix-speculative-semantics-minimc}

\autoref{fig:spec-semantics-minimc-full} shows the full speculative semantics for \MiniMC.

\begin{figure*}
  \centering
  \[
  \inferrule*[left=\textsc{McSpec\_Skip}]
    { \highlight{\mcfetch{\m}{\pc}{\cskip}} }
    { \runningstate{\scfg{\cfg{\pc}{\rho}{\mu}{\stk}}{\mfalse}{\ms}}
      \specstep{\bullet}{\bullet}
      \runningstate{\scfg{\cfg{\pc+1}{\rho}{\mu}{\stk}}{\mfalse}{\ms}}
    }
  \]\[
  \inferrule*[left=\textsc{McSpec\_Asgn}]
    { \highlight{\mcfetch{\m}{\pc}{\casgn{x}{e}}} \quad \eval{e}{\rho} = v }
    { \runningstate{\scfg{\cfg{\pc}{\rho}{\mu}{\stk}}{\mfalse}{\ms}}
      \specstep{\bullet}{\bullet}
      \runningstate{\scfg{\cfg{\pc+1}{\subst{x}{v}{\rho}}{\mu}{\stk}}{\mfalse}{\ms}}
    }
  \]\[
  \inferrule*[left=\textsc{McSpec\_Div}]
    { \highlight{\mcfetch{\m}{\pc}{\cdiv{x}{e_1}{e_2}}} \quad \eval{e_1}{\rho} = n_1 \quad \eval{e_2}{\rho} = n_2 \quad
    n = {\begin{cases} n_1/n_2 & \text{if } n_2 \neq 0 \\ \UV & \text{if } n_2 = 0 \end{cases}}
    }
    { \runningstate{\scfg{\cfg{\pc}{\rho}{\mu}{\stk}}{\mfalse}{\ms}}
      \specstep{\odiv{n_1}{n_2}}{\bullet}
      \runningstate{\scfg{\cfg{\pc+1}{\subst{x}{n}{\rho}}{\mu}{\stk}}{\mfalse}{\ms}}
    }
  \]\[
  \inferrule*[left=\textsc{McSpec\_Branch}]
  { \substack{
    \highlight{\mcfetch{\m}{\pc}{\cbranch{e}{l}}} \quad \eval{e}{\rho} = n \quad b = (n \neq 0) \\
     \pc' = \ite{b'}{\highlight{l}}{\pc+1} \quad \ms' = \ms \lor (b \neq b')
   } }
  { \runningstate{\scfg{\cfg{\pc}{\rho}{\mu}{\stk}}{\mfalse}{\ms}}
    \specstep{\obranch{b}}{\dbranch{b'}}
    \runningstate{\scfg{\cfg{\pc'}{\rho}{\mu}{\stk}}{\mfalse}{\ms'}}
  }
  \]\[
  \inferrule*[left=\textsc{McSpec\_Jump}]
  { \highlight{\mcfetch{\m}{\pc}{\cjump{l}} } }
  { \runningstate{\scfg{\cfg{\pc}{\rho}{\mu}{\stk}}{\mfalse}{\ms}}
    \specstep{\bullet}{\bullet}
    \runningstate{\scfg{\cfg{\highlight{l}}{\rho}{\mu}{\stk}}{\mfalse}{\ms}} }
  \]\[
  \inferrule*[left=\textsc{McSpec\_Load}]
  { \highlight{\mcfetch{\m}{\pc}{\cload{x}{e}}} \quad
    \eval{e}{\rho} = n \quad \mu [n] = v
  }
  { \runningstate{\scfg{\cfg{\pc}{\rho}{\mu}{\stk}}{\mfalse}{\ms}}
    \specstep{\oload{n}}{\bullet}
    \runningstate{\scfg{\cfg{\pc+1}{\subst{x}{v}{\rho}}{\mu}{\stk}}{\mfalse}{\ms}}
  }
  \]\[
  \inferrule*[left=\textsc{McSpec\_Store}]
  {
    \highlight{\mcfetch{\m}{\pc}{\cstore{e}{e'}}} \quad \eval{e}{\rho} = n \quad \eval{e'}{\rho} = v
  }
  {
    \runningstate{\scfg{\cfg{\pc}{\rho}{\mu}{\stk}}{\mfalse}{\ms}}
    \specstep{\ostore{n}}{\bullet}
    \runningstate{\scfg{\cfg{\pc+1}{\rho}{\subst{n}{v}{\mu}}{\stk}}{\mfalse}{\ms}}
  }
  \]\[
  \inferrule*[left=\textsc{McSpec\_Call}]
  {
    \highlight{\mcfetch{\m}{\pc}{\ccall{e}}} \quad \eval{e}{\rho} = l \quad \ms' = \ms \lor \pc' \neq \highlight{l}
  }
  {
    \runningstate{\scfg{\cfg{\pc}{\rho}{\mu}{\stk}}{\mfalse}{\ms}}
    \specstep{\ocall{l}}{\dcall{\pc'}}
    \runningstate{\scfg{\cfg{\pc'}{\rho}{\mu}{(\pc+1)::\stk}}{\mtrue}{\ms'}}
  }
  \]\[
  \inferrule*[left=\textsc{McSpec\_Peek}]
  { \highlight{\mcfetch{\m}{\pc}{\cpeek{x}}} \quad
    v = {\begin{cases} \pc' & \text{if } \stk = \pc' :: \stk' \\ \UV & \text{if } \stk = \epsilon \end{cases}}
  }
  { \runningstate{\scfg{\cfg{\pc}{\rho}{\mu}{\stk}}{\mfalse}{\ms}}
    \specstep{\bullet}{\bullet}
    \runningstate{\scfg{\cfg{\pc+1}{\subst{x}{v}{\rho}}{\mu}{\stk}}{\mfalse}{\ms}}
  }
  \]\[
  \inferrule*[left=\textsc{McSpec\_Ret}]
  {
    \highlight{\mcfetch{\m}{\pc}{\cret}} \quad \wfret{\pc''} \quad \ms' = \ms \lor \pc' \neq \pc''
  }
  {
    \runningstate{\scfg{\cfg{\pc}{\rho}{\mu}{\pc' :: \stk}}{\mfalse}{\ms}}
    \specstep{\bullet}{\dret{\pc''}}
    \runningstate{\scfg{\cfg{\pc''}{\rho}{\mu}{\stk}}{\mfalse}{\ms'}}
  }
  \]\[
  \inferrule*[left=\textsc{McSpec\_Term}]
  {
    \highlight{\mcfetch{\m}{\pc}{\cret}}
  }
  {
    \runningstate{\scfg{\cfg{\pc}{\rho}{\mu}{\emptylist}}{\mfalse}{\ms}}
    \specstep{\bullet}{\bullet}
    \termstate
  }
  \]\[
  \inferrule*[left=\textsc{McSpec\_CTarget}]
  {
  \highlight{\mcfetch{\m}{\pc}{\cctarget}}
  }
  {
    \runningstate{\scfg{\cfg{\pc}{\rho}{\mu}{\stk}}{\ct}{\ms}}
    \specstep{\bullet}{\bullet}
    \runningstate{\scfg{\cfg{\pc+1}{\rho}{\mu}{\stk}}{\mfalse}{\ms}}
  }
  \]\[
  \inferrule*[left=\textsc{McSpec\_Fault}]
  { \highlight{\m \texttt{[} \pc \texttt{]} \neq \cctarget} }
  {
    \runningstate{\scfg{\cfg{\pc}{\rho}{\mu}{\stk}}{\mtrue}{\ms}}
    \specstep{\bullet}{\bullet}
    \faultstate
  }
  \]
  \caption{Speculative semantics for \MiniMC}
  \label{fig:spec-semantics-minimc-full}
  \vspace{-1em}
\end{figure*}

\subsection{Linearization Transformation Function}
\label{sec:appendix-linearization-transformation}

\autoref{fig:linearization-transformation} gives the definition of the linearization
transformation function $\mc{\cdot}$ on \MiniMIR expressions and instructions; all instructions not shown are transformed to themselves (e.g. $\mc{\cskip} = \cskip$).

\begin{figure*}[t]
\centering
\begin{minipage}[t]{0.46\textwidth}
\[
\begin{aligned}
\mc{n} &= n \\
\mc{\efp(l, o)} &= \begin{array}[t]{@{}l@{}}
                   (\sum_{i=0}^{l-1} \lengthof{\prog[i]}) + o + \lengthof{\mu} \\
                   \text{where $\prog$ is the \MiniMIR program being} \\
                   \text{linearized and $\mu$ is its data memory.} \\
                   \end{array} \\
\mc{x} &= x \\
\mc{\op(e_1, \ldots, e_n)} &= \op(\mc{e_1}, \ldots, \mc{e_n}) \\
\mc{\ccond{b}{e_1}{e_2}} &= \ccond{\mc{b}}{\mc{e_1}}{\mc{e_2}}
\end{aligned}
\]
\end{minipage}
\hfill
\begin{minipage}[t]{0.46\textwidth}
\[
\begin{aligned}
\mc{\casgn{x}{e}} &= \casgn{x}{\mc{e}} \\
\mc{\cdiv{x}{e_1}{e_2}} &= \cdiv{x}{\mc{e_1}}{\mc{e_2}} \\
\mc{\cbranch{e}{l}} &= \begin{array}[t]{@{}l@{}}
                       \cbranch{\mc{e}}{l^\star} \\
                       \text{where } \mc{\efp(l, 0)} = l^\star
                       \end{array} \\
\mc{\cjump{l}} &= \begin{array}[t]{@{}l@{}}
                       \cjump{l^\star} \\
                       \text{where } \mc{\efp(l, 0)} = l^\star
                       \end{array} \\
\mc{\cload{x}{e}} &= \cload{x}{\mc{e}} \\
\mc{\cstore{e}{e'}} &= \cstore{\mc{e}}{\mc{e'}} \\
\mc{\ccall{e}} &= \ccall{\mc{e}} \\
\end{aligned}
\]
\end{minipage}
\caption{Linearization transformation}
\label{fig:linearization-transformation}
\end{figure*}

\else
\fi 

\ifieee
\bibliographystyle{abbrvnaturl}
\else 
\ifcamera
\bibliographystyle{ACM-Reference-Format}
\citestyle{acmauthoryear}   
\else
\bibliographystyle{abbrvnaturl}
\fi

\fi 

\bibliography{spec}

\end{document}